\documentclass{LMCS}

\usepackage{epsfig}
\usepackage{amssymb}
\usepackage{amsmath}
\usepackage{verbatim}
\usepackage{latexsym} 
\usepackage{enumerate}

\def\doi{6 (4:4) 2010}
\lmcsheading%
{\doi}
{1--27}
{}
{}
{Jan.~31, 2010}
{Oct.~27, 2010}
{}   

\begin{document} 

%%%%%%%%%%%%%%%%%%%%%%%%%%%%%%%%% 
%  Environments 
 
\newtheorem{theorem}{Theorem}[section]
\newtheorem{corollary}[theorem]{Corollary} 
\newtheorem{problem}{Problem} 
\newtheorem{lemma}[theorem]{Lemma} 
\newtheorem{remark}[theorem]{Remark} 
\newtheorem{observ}[theorem]{Observation} 
\newtheorem{defin}[theorem]{Definition} 
\newtheorem{example}[theorem]{Example} 
\newcommand{\PR}{{\bf Proof:\ }}             % beginning of a proof 
\def\EPR{\hfill $\Box$\linebreak\vskip2mm}  % end of a proof 
 
%%%%%%%%%%%%%%%%%%%%%%%%%%%%%%%%%%% 
%  Operators 
 
\def\Pol{{\sf Pol}} 
\def\Polo{{\sf Pol}_1\;}
\def\PPol{{\sf pPol\;}} 
\def\Inv{{\sf Inv}} 
\def\Clo{{\sf Clo}\;} 
\def\Con{{\sf Eqv}} 
\def\Sg{{\sf Sg}} 
\def\Cg{{\sf Eg}} 
\def\Un{{\sf Un}} 
\def\Str{{\sf Str}} 
\def\bin{{\sf bin}}
 
%%%%%%%%%%%%%%%%%%%%%%%%%%%%%%%%%%  
%  Operations  
  
\let\cd=\cdot  
\let\meet=\wedge  
\let\join=\vee  
\let\tm=\times  
  
%%%%%%%%%%%%%%%%%%%%%%%%%%%%%%%%%%%  
%  Accents and so  
  
\let\wh=\widehat  
\def\ox{\ov x}  
\def\oy{\ov y}  
\def\oz{\ov z}  
\def\of{\ov f}  
\def\oa{\ov a}  
\def\ob{\ov b}  
\def\oc{\ov c} 
\let\un=\underline  
\let\ov=\overline  
\let\cc=\circ  
\let\rb=\diamond  
\def\oz{{\rm [}}
\def\sz{{\rm ]}}
\def\os{{\rm (}}
\def\ss{{\rm )}}
  
%%%%%%%%%%%%%%%%%%%%%%%%%%%%%%%%%%%%%%%%%%  
%  Abbreviations for algebras and clones  
  
\def\B{\mathcal{B}}  
\def\P{\mathcal{P}}  
\def\zL{\mathbb{L}}  
\def\zD{\mathbb{D}}  
\def\zE{\mathbb{E}}  
\def\zG{\mathbb{G}}  
\def\zA{\mathbb{A}}  
\def\zB{\mathbb{B}}  
\def\zC{\mathbb{C}}  
\def\zM{\mathbb{M}}  
\def\zR{\mathbb{R}}  
\def\zS{\mathbb{S}}  
\def\zT{\mathbb{T}}  
\def\zN{\mathbb{N}}  
\def\zQ{\mathbb{Q}}  
\def\zW{\mathbb{W}}  
\def\C{\mathbf{C}}  
\def\M{\mathbf{M}}  
\def\E{\mathbf{E}}  
\def\N{\mathbf{N}}  
\def\O{\mathbf{O}}
\def\bN{\mathbf{N}}
\def\bX{\mathbf{X}}  
\def\GF{{\rm GF}}  
\def\Cc{\mathcal{C}}  
\def\cF{\mathcal{F}}
\def\cA{\mathcal{A}}  
\def\cG{\mathcal{G}}
\def\cH{\mathcal{H}}
\def\cL{\mathcal{L}}
\def\cP{\mathcal{P}}  
\def\cS{\mathcal{S}}  
\def\cT{\mathcal{T}}  
\def\tB{{\widetilde B}}  
\def\tD{{\widetilde D}}  
\def\ttB{{\widetilde{\widetilde B}}}  
\def\tba{{\tilde\ba}}  
\def\bS{\mathbf{S}}
\def\bT{\mathbf{T}}
  
%%%%%%%%%%%%%%%%%%%%%%%%%%%%%%%%%%%%%%%%  
%  Vectors  
  
\def\ba{{\bf a}}  
\def\bb{{\bf b}}  
\def\bc{{\bf c}}  
\def\bd{{\bf d}}  
\def\be{{\bf e}}  
\def\bbf{{\bf f}}  
\def\bg{{\bf g}}
\def\bh{{\bf h}}
\def\bs{{\bf s}}  
\def\bm{{\bf m}}
\def\br{{\bf r}}
\def\bt{{\bf t}}  
\def\bu{{\bf u}}  
\def\bv{{\bf v}}  
\def\bx{{\bf x}}  
\def\by{{\bf y}}  
\def\bw{{\bf w}}  
\def\bz{{\bf z}}  
\def\oal{{\ov\al}}  
\def\obeta{{\ov\beta}}  
\def\og{{\ov\gm}}  
\def\odl{{\ov\dl}}  
\def\oep{{\ov\varepsilon}}  
\def\ovm{{\ov\mu}}  
\def\ozero{{\ov0}}  
  
%%%%%%%%%%%%%%%%%%%%%%%%%%%%%%%%%%%%%%%%  
%  Constraint satisfaction Problem  
  
\def\CSP{{\rm CSP}} 
\def\CCSP{{\rm CCSP}} 
\def\ECCSP{{\rm ECCSP}} 
\def\NCSP{{\rm \#CSP}}
\def\NCCSP{{\rm \#CCSP}}
\def\FP{{\rm FP}} 
\def\CNF{{\rm CNF}} 
 
%%%%%%%%%%%%%%%%%%%%%%%%%%%%%%%%%%%%%%%%  
%  Mathematical abbreviations  
  
\let\sse=\subseteq  
\def\ang#1{\langle #1 \rangle}  
\def\dang#1{\ang{\ang{#1}}}  
\def\vc#1#2{#1 _1\zd #1 _{#2}}  
\def\zd{,\ldots,}  
\let\bks=\backslash  
\def\red#1{\vrule height7pt depth7pt width.4pt  
     \lower6pt\hbox{$\scriptstyle #1$}}  
\def\fac#1{/\lower4pt\hbox{$#1$}}  
\def\me{\stackrel{\mu}{\eq}}  
\def\nme{\stackrel{\mu}{\not\eq}}  
\def\eqc#1{\stackrel{#1}{\eq}}  
\def\sol{\stackrel{s}{\sim}}  
\def\choice#1#2{\arraycolsep0pt  
\left(\begin{array}{c} #1\\ #2 \end{array}\right)} 
\def\cl#1#2{  
\left(\begin{array}{c} #1\\ #2 \end{array}\right)}  
\def\cll#1#2#3{  
\left(\begin{array}{c} #1\\ #2\\ #3 \end{array}\right)}  
\def\clll#1#2#3#4{
\begin{array}{c} #1\\ #2\\ #3 \\#4\end{array}}  
\def\pr{{\rm pr}}  
\let\upr=\uparrow  
\def\ua#1{\hskip-1.7mm\uparrow^{#1}} 
\def\sua#1{\hskip-0.2mm\scriptsize\uparrow^{#1}} 
\def\lcm{{\rm lcm}}  
\def\w{$\wedge$}
\let\ex=\exists
\def\rel{R}
\def\relo{Q}
\def\MVM{{\sf MVM}}

%%%%%%%%%%%%%%%%%%%%%%%%%%%%%%%%%%%%%%%%%%%  
%  Other abbreviations  
  
\def\lb{$\linebreak$}  
\def\sl{\mbox{[}}
\def\sr{\mbox{]}}
  
%%%%%%%%%%%%%%%%%%%%%%%%%%%%%%%%%%%%%%%%%%  
%  Greek symbols  
  
\let\al=\alpha  
\let\gm=\gamma  
\let\dl=\delta
\let\ve=\varepsilon  
\let\kp=\kappa
\let\ld=\lambda  
\let\om=\omega  
\let\vf=\varphi  
\let\vr=\varrho  
\let\th=\theta  
\let\sg=\sigma  
\let\Gm=\Gamma  
\let\Dl=\Delta  
\let\Nb=\bigtriangledown
  
%%%%%%%%%%%%%%%%%%%%%%%%%%%%%%%%%%%%%%%%%%%%%%%%%  
%  Fonts and special symbols  
  
\font\tengoth=eufm10 scaled 1200  
\font\sixgoth=eufm6  
\def\goth{\fam12}  
\textfont12=\tengoth \scriptfont12=\sixgoth \scriptscriptfont12=\sixgoth  
\font\tenbur=msbm10  
\font\eightbur=msbm8  
\def\bur{\fam13}  
\textfont11=\tenbur \scriptfont11=\eightbur \scriptscriptfont11=\eightbur  
\font\twelvebur=msbm10 scaled 1200  
\textfont13=\twelvebur \scriptfont13=\tenbur  
\scriptscriptfont13=\eightbur  
\mathchardef\nat="0B4E  
\mathchardef\eps="0D3F

\newcommand{\relabel}[1]{\setcounter{theorem}{\ref{#1}}\addtocounter{theorem}{-1}}

%%%%%%%%%%%%%%%%%%%%%%%%%%%%%%%%%%%%%%%%%%%%%%%%%%%%%%%%%%%%%%%%
%%%%%%%%%%%%%%%%%%%%%%%%%%%%%%%%%%%%%%%%%%%%%%%%%%%%%%%%%%%%%%%%
  
\title[The complexity of global cardinality constraints]{The complexity of global cardinality constraints} 
\author[A.~Bulatov]{Andrei A.\ Bulatov\rsuper a}
\address{{\lsuper a}School of Computing Science,
 Simon Fraser University, Burnaby, Canada}
\email{abulatov@cs.sfu.ca}
\author[D.~Marx]{D\'aniel Marx\rsuper b}
\address{{\lsuper b}Tel Aviv University, Tel Aviv, Israel}
\email{dmarx@cs.bme.hu}
%% \date{}
%% \thispagestyle{empty}
  
\begin{abstract}
\noindent
  In a constraint satisfaction problem (CSP) the goal is to find an
  assignment of a given set of variables subject to specified
  constraints. A global cardinality constraint is an additional
  requirement that prescribes how many variables must be assigned a
  certain value. We study the complexity of the problem $\CCSP(\Gm)$,
  the constraint satisfaction problem with global cardinality
  constraints that allows only relations from the set $\Gm$. The main
  result of this paper characterizes sets $\Gm$ that give rise to
  problems solvable in polynomial time, and states that the remaining
  such problems are NP-complete. 
%% We extend the result also to the
%%  corresponding counting problem.
\end{abstract}
\maketitle  

%%%%%%%%%%%%%%%%%%%%%%%%%%%%%%%%%%%%%%%%%%%%%%%%%%%%%%%%%%%%%%%%%
%%%%%%%%%%%%%%%%%%%%%%%%%%%%%%%%%%%%%%%%%%%%%%%%%%%%%%%%%%%%%%%%%
\section{Introduction}
% CSP
In a constraint satisfaction problem (CSP) we are given a set of
variables, and the goal is to find an assignment of the variables
subject to specified constraints.  A constraint is usually
expressed as a requirement that combinations of values of a certain
(usually small) set of variables belong to a certain relation. CSPs
have been intensively studied in both theoretical and practical
perspectives. On the theoretical side the key research direction has
been the complexity of the CSP when either the interaction of sets
constraints are imposed on, that is, the hypergraph formed by these
sets, is restricted
\cite{Gottlob01:hypertree,Grohe07:other-side,Grohe06:fractional}, or
restrictions are on the type of allowed relations
\cite{Jeavons97:closure,Bulatov05:classifying,Bulatov03:conservative,%
Bulatov06:3-element,Barto08:graphs}.
In the latter direction the main focus has been on the so called
\emph{Dichotomy conjecture} \cite{Feder98:monotone} suggesting that
every CSP restricted in this way is either solvable in polynomial time
or is NP-complete.  

% cardinality constraint
This `pure' constraint satisfaction problem is sometimes not enough to
model practical problems, as some constraint that have to be satisfied
are not `local' in the sense that they cannot be viewed as applied to
only a limited number of variables. Constraints of this type are
called \emph{global}. Global constraints are very diverse, the current
Global Constraint Catalog (see {\small {\tt
http://www.emn.fr/x-info/sdemasse/gccat/}}) lists 348 types of
such constraints. In this paper we focus on \emph{global cardinality
  constraints}
\cite{Quimper04:cardinality,Bourdais03:hibiscus,Gomes04:cardinality}. A
global cardinality constraint $\pi$ is specified for a set of values
$D$ and a set of variables $V$, and is given by a mapping
$\pi:D\to\nat$ that assigns a natural number to each element of $D$
such that $\sum_{a\in D}\pi(a)=|V|$. An assignment of variables $V$
satisfies $\pi$ if for each $a\in D$ the number of variables that take
value $a$ equals $\pi(a)$.  In a CSP with global cardinality
constraints, given a CSP instance and a global cardinality constraint
$\pi$, the goal is to decide if there is a solution of the
CSP instance satisfying $\pi$. The restricted class of CSPs with global
cardinality constraints such that every instance from this class 
uses only relations from a fixed set $\Gm$ of relations (such a 
set is often called a \emph{constraint language}) is denoted by
$\CCSP(\Gm)$. We consider the following problem:
Characterize constraint languages $\Gm$ such that $\CCSP(\Gm)$ 
is solvable in polynomial time.
More general versions of global cardinality 
constraints have appeared in the literature, see, e.g.\ 
\cite{Quimper04:cardinality}, where the number of variables taking 
value $a$ has to belong to a prescribed set of cardinalities
(rather than being exactly $\pi(a)$). In this paper we call the CSP 
allowing such generalized constraints \emph{extended  CSP with 
cardinality constraints}. As we discuss later, our results apply to 
this problem as well.

%history
The complexity of $\CCSP(\Gm)$ has been studied in
\cite{Creignou08:cardinality} for constraint languages $\Gm$ on a
2-element set. It was shown that $\CCSP(\Gm)$ is solvable in
polynomial time if every relation in $\Gm$ is
\emph{width-2-affine}, i.e.\ it can be expressed as the set of solutions of
a system of linear equations over a 2-element field containing at most 2
variables, or, equivalently, using the equality and disequality clauses; otherwise it is NP-complete (we assume P$\ne$NP). In the 2-element case $\CCSP(\Gm)$ is
also known as the {\sc $k$-Ones}$(\Gm)$ problem, since a global
cardinality constraint can be expressed by specifying how many ones
(the set of values is thought to be $\{0,1\}$) one wants to have among
the values of variables. The parameterized complexity of {\sc
  $k$-Ones}$(\Gm)$ has also been studied \cite{Marx05:parametrized},
where $k$ is used as a parameter.  

In the case of a 2-element domain, the polynomial classes can be handled
by a standard application of dynamic programming. Suppose that the
instance is given by a set of unary clauses and binary
equality/disequality clauses. Consider the graph formed by the binary
clauses. There are at most two possible assignments for each 
connected component of the graph: setting the value of a variable
uniquely determines the values of all the other variables in the
component. Thus the problem is to select one of the two assignments
for each component. Trying all possibilities would be exponential in
the number of components. Instead, for $i=1,2,\dots$, we compute the
set $S_i$ of all possible pairs $(x,y)$ such that there is a
partial solution on the first $i$ components containing exactly $x$
zeroes and exactly $y$ ones. It is not difficult to see that $S_{i+1}$
can be computed if $S_i$ is already known. 

We generalize the results of \cite{Creignou08:cardinality} for
arbitrary finite sets and arbitrary constraint languages.  As usual, the
characterization for arbitrary finite sets is significantly more
complex and technical than for the 2-element set. As a straightforward
generalization of the 2-element case, we can observe that the problem
is polynomial-time solvable if every relation can be expressed by
binary mappings. In this case, setting a single value in a component
uniquely determines all the values in the component. Therefore, if the
domain is $D$, then there are at most $|D|$ possible assignments in each
component, and the same dynamic programming technique can be applied
(but this time the set $S_i$ contains $|D|$-tuples instead of pairs).

One might be tempted to guess that the class described in the previous
paragraph is the only class where $\CCSP$ is polynomial-time solvable.
However, it turns out that there are more general tractable classes.
First, suppose that the domain is partitioned into equivalence
classes, and the binary constraints are mappings between the
equivalence classes. This means that the values in the same
equivalence class are completely interchangeable. Thus it is
sufficient to keep one representative from each class, and then the
problem can be solved by the algorithm sketched in the previous
paragraph. Again, one might believe that this construction gives all
the tractable classes, but the following example shows that it does
not cover all the tractable cases.

\begin{example}\label{exa:subcomponent}
  Let $R=\{(1,2,3),(1,4,5),(a,b,c),(d,e,c)\}$. We claim that
  $\CCSP(\{R\})$ is poly\-no\-mi\-al-time solvable. Consider the graph on
  the variables where two variables are connected if and only if they
  appear together in a constraint. As before, for each component, we
  compute a set containing all possible cardinality vectors, and then
  use dynamic programming. In each component, we have to consider only
  two cases: either every variable is in $\{1,2,3,4,5\}$ or every
  variable is in $\{a,b,c,d,e\}$. If every variable of component $K$
  is in $\{1,2,3,4,5\}$, then $R$ can be expressed by the unary
  constant relation $\{1\}$, and the binary relation
  $R'=\{(2,3),(4,5)\}$. The binary relations partition component $K$
  into sub-components $K_1$, $\dots$, $K_t$. Since $R'$ is a mapping,
  there are at most 2 possible assignments for each sub-component.
  Thus we can use dynamic programming to compute the set of all
  possible cardinality vectors on $K$ that use only the values in
  $\{1,2,3,4,5\}$. If every variable of $K$ is in $\{a,b,c,d,e\}$,
  then $R$ can be expressed as the unary constant relation $\{c\}$ and
  the binary relation $R''=\{(a,b),(d,e)\}$.  Again, binary relation
  $R''$ partitions $K$ into sub-components, and we can use dynamic
  programming on them. Observe that the sub-components formed by $R'$
  and the sub-components formed by $R''$ can be different: in the
  first case, $u$ and $v$ are adjacent if they appear in the second
  and third coordinates of a constraint, while in the second case, $u$
  and $v$ are adjacent if they appear in the first and second
  coordinates of a constraint.
\end{example}

It is not difficult to make Example~\ref{exa:subcomponent} more
complicated in such a way that we have to look at sub-subcomponents
and perform multiple levels of dynamic programming. This suggests that
it would be difficult to characterize the tractable relations in a
simple combinatorial way.

We give two characterizations of finite $\CCSP$, one more along the line of the usual approach to the CSP, using polymorphisms, and another more 
combinatorial one. The latter is more technical, but it is much more suitable for algorithms.
 
A polymorphism of a constraint language is an operation that preserves
every relation from the language. The types of polymorphisms we need
here are quite common and have appeared in the literature many times.
A ternary operation $m$ satisfying the equations
$m(x,x,y)=m(x,y,x)=m(y,x,x)=x$ is said to be \emph{majority}, and a ternary
operation $h$ satisfying $h(x,y,y)=h(y,y,x)=x$ is said to be \emph{Mal'tsev}.
An operation is \emph{conservative} if it always takes a value equal to one
(not necessarily the same one) of its arguments.
 
\begin{theorem}\label{the:main-poly}
For a constraint language $\Gm$, the problem $\CCSP(\Gm)$ is polynomial time solvable if and only if $\Gm$ has a majority polymorphism and a conservative Mal'tsev polymorphism.
Otherwise it is NP-complete.
\end{theorem}

Observe that for constraint languages over a 2-element domain,
Theorem~\ref{the:main-poly} implies the characterization of Creignou et
al.~\cite{Creignou08:cardinality}. Width-2 affine is equivalent to
affine and bijunctive (definable in 2SAT), and over a 2-element domain, affine is
equivalent to having a conservative Mal'tsev polymorphism and
bijunctive is equivalent to having a majority polymorphism.

The second characterization uses logical
definability. The right generalization of mappings is given by the notion of rectangularity. A binary relation $\rel$ is called \emph{rectangular} if $(a,c),(a,d),(b,d)\in\rel$ implies $(b,c)\in\rel$. We say that a pair of equivalence relations $\alpha$ and
$\beta$ over the same domain {\em cross}, if there is an $\al$-class
$C$ and a $\beta$-class $D$ such that none of $C\setminus D$,
$C\cap D$, and $D\setminus C$ is empty. A relation is {\em
2-decomposable} if it can be expressed as a conjunction of binary
relations. We denote by $\dang\Gm$ the set of all relations that are
primitive positive definable in $\Gm$. A constraint language is said to be \emph{non-crossing decomposable}
if every relation from $\dang\Gm$ is 2-decomposable, every binary
relation from $\dang\Gm$ is rectangular, and no pair of equivalence
relations from $\dang\Gm$ cross. For detailed definitions and discussion see
Section~\ref{sec:preliminaries}. 

\begin{theorem}\label{the:main-poly2}
For a constraint language $\Gm$, the problem $\CCSP(\Gm)$ is polynomial time solvable if and only if $\Gm$ is non-crossing decomposable.
Otherwise it is NP-complete.
\end{theorem}

The equivalence of the two characterizations will be proved in Section~\ref{sec:equiv-char}.

%counting problem
Following \cite{Creignou08:cardinality}, we also study the counting problem
$\NCCSP(\Gm)$ corresponding to $\CCSP(\Gm)$, in which the objective is to find the
number of solutions of a CSP instance that satisfy a global
cardinality constraint specified. 
Creignou et al.~\cite{Creignou08:cardinality} proved that if $\Gm$ is
a constraint language on a 2-element set, say, $\{0,1\}$, then
$\NCCSP(\Gm)$ are solvable in polynomial time exactly when $\CCSP(\Gm)$ is, that is, if
every relation from $\Gm$ is width-2-affine. Otherwise it is $\textup{P}^{\#\textup{P}}$-complete.

We prove that in the general case as well, $\NCCSP(\Gm)$ is polynomial time solvable if and only if $\CCSP(\Gm)$ is. However, in this paper we do not prove a
complexity dichotomy, as we do not determine the exact complexity of
the hard counting problems. All such problems are NP-hard as Theorems~\ref{the:main-poly} and~\ref{the:main-poly2} show; and we do not claim that the NP-hard cases are actually
$\textup{P}^{\#\textup{P}}$-hard.

\begin{theorem}\label{the:counting-main-poly}
For a constraint language $\Gm$, the problem $\NCCSP(\Gm)$ is polynomial time solvable if and only if $\Gm$ has a majority polymorphism and a conservative Mal'tsev polymorphism; or, equivalently, if and only if $\Gm$ is non-crossing decomposable.
Otherwise it is NP-hard.
\end{theorem}

We also consider the so called \emph{meta-problem} for $\CCSP(\Gm)$ and $\NCCSP(\Gm)$:
Suppose set $D$ is fixed. Given a finite constraint language $\Gm$ on
$D$, decide whether or not $\CCSP(\Gm)$ (and $\NCCSP(\Gm)$) is solvable in polynomial
time. By Theorems~\ref{the:main-poly} and~\ref{the:counting-main-poly} it suffices to check if $\Gm$ has a majority and a conservative Mal'tsev polymorphism. Since the set $D$ is fixed, this can be done by checking, for each possible ternary function with the required properties, whether or not it is a polymorphism of $\Gm$.  To check if a ternary operation $f$ is a polymorphism of $\Gm$ one just needs for each relation $\rel\in\Gm$
to apply $f$ to every triple of tuples in $\rel$. This can be done in a time cubic in the total size of
relations in $\Gm$.

\begin{theorem}\label{the:meta-problem}
Let $D$ be a finite set. The meta-problem for $\CCSP(\Gm)$ and $\NCCSP(\Gm)$ is polynomial time solvable.
\end{theorem}

Note that all the results use the assumption that the set $D$ is fixed 
(although the complexity of algorithms does not depend on a particular
constraint language). Without this assumption the algorithms given in 
the paper become exponential time, and Theorem~\ref{the:meta-problem} 
does not answer if the meta
problem is polynomial time solvable if the set $D$ is not fixed, and
is a part of the input. The algorithm sketched above is then super-exponential.

%%%%%%%%%%%%%%%%%%%%%%%%%%%%%%%%%%%%%%%%%%%%%%%%%%%%%%%%%%
%%%%%%%%%%%%%%%%%%%%%%%%%%%%%%%%%%%%%%%%%%%%%%%%%%%%%%%%%%
\section{Preliminaries}\label{sec:preliminaries}

%%%%%%%%%%%%%%%%%%%%%%%%%%%%%%%%%%%%%%%%%%%%%%%%%%%%%%%%%%
%%%%%%%%%%%%%%%%%%%%%%%%%%%%%%%%%%%%%%%%%%%%%%%%%%%%%%%%%%
\subsection*{Relations and constraint languages}

Our notation concerning tuples and relations is fairly standard.
The
set of all tuples of elements from a set $D$ is denoted by $D^n$. We
denote tuples in boldface, e.g., $\ba$, and their components by
$\ba[1],\ba[2],\ldots$. An $n$-ary \emph{relation} on set $D$ is any
subset of $D^n$. Sometimes we use instead of relation
$\rel$ the corresponding predicate $\rel(\vc xn)$. A set of relations on $D$ is called a \emph{constraint language} over $D$. 
%% Using predicates we can
%% \emph{express} or \emph{define} relations through other relations by
%% means of logical formulas. 

For a subset $I=\{\vc ik\}\sse\{1\zd n\}$ with
$i_1<\ldots<i_k$ and an $n$-tuple $\ba$, by $\pr_I\ba$ we denote the
\emph{projection of} $\ba$ \emph{onto} $I$, the $k$-tuple
$(\ba[i_1]\zd\ba[i_k])$. The \emph{projection} $\pr_I\rel$ of $\rel$
is the $k$-ary relation $\{\pr_I\ba\mid\ba\in\rel\}$. Sometimes we need to emphasize that the unary projections
$\pr_1\rel$, $\pr_2\rel$ of a binary relation $\rel$ are 
sets $A$ and $B$. We denote this by $\rel\sse A\tm B$.

Pairs from equivalence relations play a special role, so such pairs
will be denoted by, e.g., $\ang{a,b}$. If $\al$ is an equivalence
relation on a set $D$ then $D\fac\al$ denotes the set of
$\al$-classes, and $a^\al$ for $a\in D$ denotes the $\al$-class
containing $a$. We say that the equivalence relation $\al$ on a set
$D$ is {\em trivial} if $D$ is the only $\al$ class.

%%%%%%%%%%%%%%%%%%%%%%%%%%%%%%%%%%%%%%%%%%%%%%%%%%%%%%%%%%
\subsection*{Constraint Satisfaction Problem with cardinality constraints}

Let $D$ be a finite set (throughout the paper we assume it fixed) 
and $\Gm$ a constraint language over $D$. An
instance of the \emph{Constraint Satisfaction Problem} (CSP for short)
$\CSP(\Gm)$ is a pair $\cP=(V,\Cc)$, where $V$ is a finite set of
\emph{variables} and $\Cc$ is a set of \emph{constraints}. Every
constraint is a pair $C=\ang{\bs,\rel}$ consisting of an $n_C$-tuple $\bs$ of
variables, called the \emph{constraint scope} and an $n_C$-ary
relation $\rel\in\Gm$, called the \emph{constraint relation}. A
solution of $\cP$ is a mapping $\vf\colon V\to D$ such that for every
constraint $C=\ang{\bs,\rel}$ the tuple $\vf(\bs)$ belongs to $\rel$. 

A \emph{global cardinality constraint} for a CSP instance $\cP$ is a mapping
$\pi:D\to\nat$ with $\sum_{a\in D}\pi(a)=|V|$. A solution $\vf$ of
$\cP$ satisfies the cardinality 
constraint $\pi$ if the number of variables mapped to each $a\in D$ equals
$\pi(a)$. The variant of $\CSP(\Gm)$ allowing global
cardinality constraints will be denoted by
$\CCSP(\Gm)$; the question \begin{comment}in this problem
\end{comment}
is, given an instance $\cP$ and a cardinality constraint $\pi$,
whether there is a solution of $\cP$ satisfying $\pi$. 

\begin{example}\label{exa:k-ones}
If $\Gm$ is a constraint language on the 2-element set $\{0,1\}$ then
to specify a global cardinality constraint it suffices to specify the
number of ones we want to have in a solution. This problem is also
known as the {\sc $k$-Ones}$(\Gm)$ problem
\cite{Creignou08:cardinality}. 
\end{example}

\begin{example}\label{exa:coloring}
  Let $\Gm_\textsc{3-Col}$ be the constraint language on $D=\{0,1,2\}$
  containing only the binary disequality relation $\ne$. It is
  straightforward that $\CSP(\Gm_\textsc{3-Col})$ is equivalent to the
  {\sc Graph 3-Colorability} problem. Therefore
  $\CCSP(\Gm_\textsc{3-Col})$ is equivalent to the {\sc Graph
    3-Colorability} problem in which the question is whether there is
  a coloring with a prescribed number of vertices colored each color.
\end{example}

Sometimes it is convenient to use arithmetic operations on cardinality
constraints. Let $\pi,\pi':D\to\nat$ be cardinality constraints on a
set $D$, and $c\in\nat$. Then $\pi+\pi'$ and $c\pi$ denote cardinality
constraints given by $(\pi+\pi')(a)=\pi(a)+\pi'(a)$ and
$(c\pi)(a)=c\cdot\pi(a)$, respectively, for any $a\in D$. Furthermore,
we extend addition to sets $\Pi$, $\Pi'$ of cardinality vectors in a
convolution sense: $\Pi+\Pi'$ is defined as $\{\pi+\pi'\mid \pi\in
\Pi, \pi'\in \Pi'\}$.

It is possible to consider an even more general CSP with global
cardinality constraints, in which every instance of $\CSP(\Gm)$ is
accompanied with a set of global cardinality constraints, and the
question is whether or not there exists a solution of the CSP instance
that satisfies one of the cardinality constraints. Sometimes such a
set of cardinality constraints can be represented concisely, for
example, all constraints $\pi$ with $\pi(a)=k$. We denote such
\emph{extended CSP with global cardinality constraints} corresponding
to a constraint language $\Gm$ by $\ECCSP(\Gm)$.  

\begin{example}\label{exa:extended-coloring}
The problem $\ECCSP(\Gm_\textsc{3-Col})$ admits a wide variety of questions, e.g.\, whether a given graph admits a 3-coloring with 25 vertices colored 0, and odd number of vertices colored 1.
\end{example}

However, in our
setting (as $|D|$ is a fixed constant and we are investigating
polynomial-time solvability) the extended problems are not very
interesting from the complexity point of view. 

\begin{lemma}\label{lem:eccsp-to-ccsp}
For any constraint language $\Gm$ the problem $\ECCSP(\Gm)$ is Turing
reducible to $\CCSP(\Gm)$. 
\end{lemma}

\begin{proof}
  Since $D$ is fixed, for any instance $\cP$ of $\CSP(\Gm)$ there are
  only polynomially many global cardinality constraints. Thus we can
  try each of the cardinality constraints given in an instance of
  $\ECCSP(\Gm)$ in turn.
\end{proof}
Note that the algorithm in Section~\ref{sec:algorithm} for $\CCSP$
actually finds the set of all feasible cardinality constraints. Thus
$\ECCSP$ can be solved in a more direct way than the reduction in
Lemma~\ref{lem:eccsp-to-ccsp}.

%%%%%%%%%%%%%%%%%%%%%%%%%%%%%%%%%%%%%%%%%%%%%%%%%%%%%%%%%%
\subsection*{Primitive positive definitions and polymorphisms}

Let $\Gm$ be a constraint language on a set $D$. A relation $\rel$ is
\emph{primitive positive} \emph{(pp-)} \emph{definable} in
$\Gm$ if it can be expressed using (a) relations from $\Gm$, (b)
conjunction, (c) existential quantifiers, and (d) the binary equality
relations (see, e.g.\ \cite{Denecke-Wismath02}). The set of all relations pp-definable in $\Gm$ will be
denoted by $\dang\Gm$. 

\begin{example}\label{exa:product}
An important example of pp-definitions that will be used throughout
the paper is the \emph{product} of binary relations. Let $\rel,\relo$
be binary relations. Then $\rel\circ\relo$ is the binary relation
given by 
$$
(\rel\circ\relo)(x,y)=\exists z\rel(x,z)\wedge\relo(z,y).
$$
\end{example}

In this paper we will need a slightly weaker notion of
definability. We say that $\rel$ is \emph{pp-definable} in $\Gm$
\emph{without equalities} if it can be expressed using only items
(a)--(c) from above. The set of all relations pp-definable in $\Gm$
without equalities will be denoted by $\dang\Gm'$. Clearly,
$\dang\Gm'\sse\dang\Gm$. 
The two sets are different only on
relations with redundancies.
 Let $\rel$ be a (say, $n$-ary) relation. A \emph{redundancy} of
 $\rel$ is a pair $i,j$ of its coordinate positions such that, for any
 $\ba\in\rel$, $\ba[i]=\ba[j]$.  
 
 \begin{example}\label{exa:redundancy}
 In some cases if a relation $\rel$ has redundancies, the equality relation is pp-definable in $\{\rel\}$ without equalities. Let $\rel$ be a ternary relation on $D=\{0,1,2\}$ given by (triples, members of the relation, are written vertically)
 $$
 \rel=\left(\begin{array}{cccccc}
 0&0&1&1&2&2 \\ 0&0&1&1&2&2 \\ 1&2&0&2&0&1
 \end{array}\right).
 $$
Then the equality relation is expressed by $\exists z \rel(x,y,z)$.

In other cases the equality relation cannot be expressed that easily, but its restriction onto a subset of $D$ can. Let $\relo$ be a 4-ary on $D=\{0,1,2\}$ given by
$$
 \rel=\left(\begin{array}{cccccc}
 0&0&0&2&2&2 \\ 0&0&0&2&2&2 \\ 1&2&0&2&0&1 \\ 0&0&1&1&2&2
 \end{array}\right).
 $$
 Then the formula $\exists z,t \rel(x,y,z,t)$ defines the equality relation on $\{0,2\}$.
 \end{example}

\begin{lemma}\label{lem:redundancy}
For every constraint language $\Gm$, every  $\rel\in\dang\Gm$ without
redundancies belongs to $\dang\Gm'$. 
\end{lemma}

\begin{proof}
  Consider a pp-definition of $\rel$ in $\Gm$. Suppose that the
  definition contains an equality relation on the variables $x$ and
  $y$. If none of $x$ and $y$ is bound by an existential quantifier,
  then the relation $\rel$ has two coordinates that are always equal,
  i.e., $\rel$ is redundant. Thus one of the variables, say $x$, is
  bound by an existential quantifier. In this case, replacing $x$ with
  $y$ everywhere in the definition defines the same relation $\rel$
  and decreases the number of equalities used. Repeating this step, we
  can arrive to an equality-free definition of $\rel$.
\end{proof}

A \emph{polymorphism} of a (say, $n$-ary) relation $\rel$ on $D$ is a
mapping $f:D^k\to D$ for some $k$ such that for any tuples $\vc\ba
k\in\rel$ the tuple  
\[
f(\vc\ba k)=(f(\ba_1[1]\zd\ba_k[1])\zd f(\ba_1[n]\zd\ba_k[n]))
\]
belongs to $\rel$. Operation $f$ is a polymorphism of a constraint
language $\Gm$ if it is a polymorphism of every relation from
$\Gm$. There is a tight connection, a \emph{Galois correspondence}, between
polymorphisms of a constraint language and relations pp-definable in
the language, see \cite{Geiger68:closed,Bodnarchuk69:Galua1}. This
connection has been extensively exploited to study the ordinary
constraint satisfaction problems
\cite{Jeavons97:closure,Bulatov05:classifying}. Here we do
not need the full power of this Galois correspondence, we only need
the following folklore result:
 
\begin{lemma}\label{lem:Galois}
If operation $f$ is a polymorphism of a constraint language $\Gm$,
then it is also a polymorphism of any relation from $\dang\Gm$, and
therefore of any relation from $\dang\Gm'$. 
\end{lemma}
For a (say, $n$-ary) relation $\rel$ over a set $D$ and a subset
$D'\sse D$, by $\rel_{|D'}$ we denote the relation $\{(\vc an)\mid
(\vc an)\in\rel \text{ and } \vc an\in D'\}$. For a constraint
language $\Gm$ over $D$ we use $\Gm_{|D'}$ to denote the constraint
language $\{\rel_{|D'}\mid \rel\in\Gm\}$.

If $f$ is a polymorphism of a constraint language $\Gm$ over $D$ and
$D'\subset D$, then $f$ is not necessarily a polymorphism of
$\Gm_{|D'}$. However, it remains a polymorphism in the following
special case.  A $k$-ary polymorphism $f$ is {\em conservative,} if
$f(a_1,\dots, a_k)\in \{a_1,\dots,a_k\}$ for every $a_1,\dots,a_k\in
D$. It is easy to see that if $f$ is a conservative polymorphism of
$\Gm$, then $f$ is a (conservative) polymorphism of $\Gm_{|D'}$
for every $D'\subseteq D$.

Polymorphisms help to express many useful properties of relations. 
A (say, $n$-ary) relation $\rel$ is said to be \emph{2-decomposable}
if $\ba\in\rel$ if and only if, for any $i,j\in\{1\zd n\}$,
$\pr_{i,j}\ba\in\pr_{i,j}\rel$, see
\cite{Baker75:chinese-remainder,Jeavons98:consist}.  Decomposability
sometimes is a consequence of the existence of certain polymorphisms.
A ternary operation $m$ on a set $D$ is said to be a \emph{majority
  operation} if it satisfies equations $m(x,x,y)=m(x,y,x)=m(y,x,x)=x$
for all $x,y\in D$. By \cite{Baker75:chinese-remainder} if a majority
operation $m$ is a polymorphism of a constraint language $\Gm$ then
$\Gm$ is 2-decomposable. The converse is not true: there are
2-decomposable relations not having a majority polymorphism.
Furthermore, 2-decomposability is not preserved by pp-definitions,
thus we cannot expect to characterize it by polymorphisms.

\begin{example}\label{exa:2decomp}
  Consider the disequality relation $\neq$ over the set $D=\{1,2,3\}$.
  Relation $\neq$ is trivially 2-decomposable, since it is binary. Let
  $R(x,y,z)=\exists q ((x\neq q)\wedge (y\neq q)\wedge (z\neq q))$.
  The binary projections of $R$ are $D\times D$, but $R$ is not
  $D\times D\times D$: it does not allow that $x$, $y$, $z$ are all
  different.
\end{example}

A binary relation $\rel$ is said to be \emph{rectangular} if for any
$(a,c),(a,d),(b,d)\in\rel$, the pair $(b,c)$ also belongs to $\rel$.
Rectangular relations and their generalizations play a very important
role in the study of CSP
\cite{Bulatov06:simple,Bulatov07:towards,Idziak07:tractability}.
A useful way to think about binary rectangular relations is to represent 
them as \emph{thick
mappings}.  A binary relation $\rel\sse A\tm B$ is called a thick
mapping if there are equivalence relations $\al$ and $\beta$ on $A$
and $B$, respectively, and a one-to-one mapping $\vf\colon A\fac\al\to
B\fac\beta$ (thus, in particular, $|A\fac\al|=|B\fac\beta|$) such that
$(a,b)\in\rel$ if and only if $b^\beta=\vf(a^\al)$. In this case we
shall also say that $\rel$ is a \emph{thick mapping with respect to} $\al$,
\label{page:with-respect}
$\beta$, and $\vf$. Given a thick mapping $\rel$ the corresponding
equivalence relations will be denoted by $\al_\rel^1$ and
$\al_\rel^2$. Observe that $\al_\rel^1=\rel\circ\rel^{-1}$ and
 $\al_\rel^2=\rel^{-1}\circ\rel$; therefore $\al_\rel^1,\al_\rel^2\in\dang{\{\rel\}}$. 
Thick mapping $\rel$ is said to be \emph{trivial} if
both $\al_\rel^1$ and $\al_\rel^2$ are the total equivalence relations
$(\pr_1\rel)^2$ and $(\pr_2\rel)^2$. In a graph-theoretical point of
view, a thick mapping defines a bipartite graph where every connected
component is a complete bipartite graph. 

As with decomposability,
rectangularity follows from the existence of a certain polymorphism. A
ternary operation $h$ is said to be \emph{Mal'tsev} if
$h(x,x,y)=h(y,x,x)=y$ for all $x,y\in D$. The first part of following lemma is straightforward, while the second part is folklore

\begin{lemma}\label{lem:thick-mapping}
(1) Binary relation $\rel$ is a thick mapping if and only if it is rectangular.\\[1mm]
(2) If a binary relation $\rel$ has a Mal'tsev polymorphism then it is rectangular.
\end{lemma}

A ternary operation $h$ satisfying equations $h(x,x,y)=h(x,y,x)=h(y,x,x)=y$
for all $x,y\in D$ is said to be a \emph{minority operation}. Observe that every Mal'tsev operation is minority, but not the other way round.

%%%%%%%%%%%%%%%%%%%%%%%%%%%%%%%%%%%%%%%%%%%%%%%%%%%%%%%%%%
\subsection*{Consistency}

Let us fix a constraint language $\Gm$ on a set $D$ and let
$\cP=(V,\Cc)$ be an instance of $\CSP(\Gm)$. A \emph{partial solution}
of $\cP$ on a set of variables $W\sse V$ is a mapping $\psi:W\to D$ that
satisfies every constraint
$\ang{W\cap\bs,\pr_{W\cap\bs}\rel}$ where $\ang{\bs,\rel}\in\Cc$. 
Here $W\cap\bs$ denotes the
subtuple of $\bs$ consisting of those entries of $\bs$ that belong to
$W$, and we consider the coordinate positions of $\rel$ indexed by variables from $\bs$. Instance $\cP$ is said to be \emph{$k$-consistent} if for any
$k$-element set $W\sse V$ and any $v\in V\setminus W$ any partial
solution on $W$ can be extended to a partial solution on
$W\cup\{v\}$, see \cite{Jeavons98:consist}. As we only need $k=2$, all further definitions are given
under this assumption. 

Any instance $\cP=(V,\Cc)$ can be transformed to a 2-consistent
instance by means of the standard {\sc 2-Consistency} algorithm. This
algorithm works as follows. First, for each pair $v,w\in V$ it creates
a constraint $\ang{(v,w),\rel_{v,w}}$ where $\rel_{v,w}$ is the binary
relation consisting of all partial solutions $\psi$ on $\{v,w\}$,
i.e.\ $\rel_{v,w}$ includes pairs $(\psi(v),\psi(w))$. These new
constraints are added to $\Cc$, let the resulting instance be denoted
by $\cP'=(V,\Cc')$. Second, for each pair $v,w\in V$, every partial
solution $\psi\in\rel_{v,w}$, and every $u\in V\setminus\{v,w\}$, the
algorithm checks if $\psi$ can be extended to a partial solution of
$\cP'$ on $\{v,w,u\}$. If not, it updates $\cP'$ by removing $\psi$
from $\rel_{v,w}$. The algorithm repeats this step until no more
changes happen.  

\begin{lemma}\label{lem:consistency}
Let $\cP=(V,\Cc)$ be an instance of $\CSP(\Gm)$.\\
(a) The problem obtained from $\cP$ by applying {\sc 2-Consistency} is 2-consistent;\\
(b) On every step of {\sc 2-Consistency} for any pair $v,w\in V$ the
relation $\rel_{v,w}$ belongs to $\dang\Gm'$. 
\end{lemma}

\begin{proof}
(a) follows from \cite{Cooper89:kconsistency}.

(b) Since after the first phase of the algorithm every relation $\rel_{v,w}$ is an intersection of unary and binary projections of relations from $\Gm$, it belongs to $\dang\Gm'$. Then when considering a pair $v,w\in V$ and $u\in V\setminus\{v,w\}$, the relation $\rel_{v,w}$ is replaced with $\rel_{v,w}\cap\pr_{v,w}\relo$, where $\relo$ is the set of all solution of the current instance on $\{v,w,u\}$. As every relation of the current instance belongs to $\dang\Gm'$, the relation $\relo$ is pp-definable in $\Gm$ without equalities. Thus the updated relation $\rel_{v,w}$ also belongs to $\dang\Gm'$.
\end{proof}
Note that Theorem~\ref{lem:consistency}(b) implies that any
polymorphism of $\Gamma$ is also a polymorphism of every $R_{v,w}$.

If $\Gm$ has a majority polymorphism, by Theorem~3.5 of 
\cite{Jeavons98:consist}, every 2-consistent problem instance is
\emph{globally consistent}, that is every partial solution can be
extended to a global solution of the problem. In particular, if $\cP$ is 2-consistent, then for any $v,w\in V$, any pair $(a,b)\in\rel_{v,w}$ can be extended to a solution of $\cP$. The same is true for any $a\in\pr_1\rel_{v,w}$ and any $b\in\pr_2\rel_{v,w}$.

%%%%%%%%%%%%%%%%%%%%%%%%%%%%%%%%%%%%%%%%%%%%%%%%%%%%%%%%%%
%%%%%%%%%%%%%%%%%%%%%%%%%%%%%%%%%%%%%%%%%%%%%%%%%%%%%%%%%%
\section{Equivalence of the characterizations}\label{sec:equiv-char}

In this section, we prove that the two characterizations in
Theorems~\ref{the:main-poly} and \ref{the:main-poly2} are equivalent. 
Recall that two equivalence relations $\al$ and $\beta$ over the same domain \emph{cross}, if there is an $\al$-class $A$ and a $\beta$-class $B$ such that none of $A\setminus B$, $A\cap B$, and $B\setminus A$ is empty. A constraint language $\Gm$ is said to be \emph{non-crossing decomposable} if every relation from $\dang\Gm$ is 2-decomposable, every binary relation from $\dang\Gm$ is a thick mapping, and no pair of equivalence relations from $\dang\Gm$ cross.

One of the directions is easy to see. 

\begin{lemma}\label{lem:equiv0}
If $\Gm$ has a majority
polymorphism $m$ and a conservative Mal'tsev polymorphism $h$, then $\Gm$ is non-crossing decomposable.
\end{lemma}

\begin{proof}
As $m$ is a polymorphism of every relation in $\dang{\Gm}$, by \cite{Baker75:chinese-remainder} every
relation in $\dang{\Gm}$ is 2-decomposable. Similarly, every binary
relation in $\dang{\Gm}$ is invariant under the  Mal'tsev polymorphism,
which, by Lemma~\ref{lem:thick-mapping}, implies that the binary relation is a
thick mapping.

Finally, suppose that there are two equivalence relations
$\al,\beta\in\dang\Gm$ over the same domain $D'$ that cross. This means that for some
$a,b,c\in D'$ we have that $\ang{a,b}\in\al$, $\ang{b,c}\in\beta$, but
$\ang{a,b}\not\in\beta$, $\ang{b,c}\not\in\al$. Let $h$ be a Mal'tsev
polymorphism of $\Gm$, and consider $d=h(a,b,c)$. First of all, as $h$
is conservative, $d\in\{a,b,c\}$. Then, since $h$ preserves $\al$ and
$\beta$,
$$
h\left(\begin{array}{ccc} a&b&c\\ a&a&c\end{array}\right)=\cl dc\in\al,\qquad
h\left(\begin{array}{ccc} a&b&c\\ a&c&c\end{array}\right)=\cl da\in\beta,
$$
which is impossible.
\end{proof}

To prove the other direction of the equivalence, we need to construct
the two polymorphisms. The following definition will be useful for
this purpose.

\begin{defin} Given a constraint language $\Gm$, we say that $(a|bc)$
is true if $\dang\Gm$ contains an equivalence relation $\al$ with
\[
\ang{a,a},\ang{b,b} ,\ang{c,c},\ang{b,c},\ang{c,b}\in \al \text{ and }\\
\ang{a,b},\ang{a,c},\ang{b,a},\ang{c,a}\not\in \al.
\]
\end{defin}
In other words, the domain of the equivalence relation contains all
three elements, $b$ and $c$ are in the same class, but $a$ is in a
different class. 

\begin{lemma}\label{lem:triple}
Let $\Gm$ be a non-crossing decomposable constraint language over $D$.
For every $a,b,c\in D$, at
most one of $(a|bc)$, $(b|ac)$, $(c|ab)$ is true.
\end{lemma}

\begin{proof}
  Suppose that, say, both $(a|bc)$ and $(b|ac)$ are true; let
  $\al,\beta$ be the corresponding equivalence relations from
  $\dang{\Gm}$. Let $D_\al$, $D_\beta$ be the domains of $\al$ and
  $\beta$, respectively. We can consider $D_\al$ and $D_\beta$ as
  unary relations, and they are in $\dang\Gm$. Therefore,
  $D'=D_\al\cap D_\beta$ and $\al'=\al\cap(D'\times D')$ and
  $\beta'=\beta\cap(D'\times D')$ are also in $\dang\Gm$. As is easily
  seen, $\al'$ and $\beta'$ are over the same domain $D'$ and they
  cross, a contradiction.
\end{proof}

\begin{lemma}\label{lem:triple2}
  Let $\Gm$ be a non-crossing decomposable constraint language over $D$. 
  Let $R\in \dang\Gm$ be a binary relation such
  that $(a,a'),(b,b'),(c,c')\in R$, but $(p,q)\not\in \rel$ for some
  $p\in\{a,b,c\}$ and $q\in \{a',b',c'\}$. Then
\begin{gather*}
(a|bc) \iff (a'|b'c')\\
(b|ac) \iff (b'|a'c')\\
(c|ab) \iff (c'|a'b').
\end{gather*}
\end{lemma}
\begin{proof}
If, say, $a$ and $b$ are in the same $\alpha^1_{R}$ class, then $c$
has to be in a different class. Since $R$ is a thick mapping, this means that $a'$ and $b'$ are in the
same $\alpha^2_R$ class and $c'$ is in a different class. Therefore,
both $(c|ab)$ and $(c'|a'b')$ are true, and by Lemma~\ref{lem:triple},
none of the other statements can be true. Therefore, we can assume that
  $a$, $b$, $c$ are in different $\alpha^1_R$ classes, and hence $a'$,
  $b'$, $c'$ are in different $\alpha^2_R$ classes.

Suppose that $(a|bc)$ is true; let $\al\in\dang\Gm$ be the
corresponding equivalence relation. 
Consider the relation $\rel'=\al \circ \rel$, that is, $\rel'(x,y)=\exists z
(\al(x,z)\wedge\rel(z,y))$,  which has to be a thick mapping. 
Let $\beta^1_{R'}, \beta^2_{R'}$ be the
equivalence relations of $\rel'$. Observe that $\beta^1_{R'}, 
\beta^2_{R'}\in\dang{\rel'}\sse\dang\Gm$. We claim that $\ang{b',c'}\in
\beta^2_{\rel'}$ and $\ang{a',b'}\not\in \beta^2_{\rel'}$, showing
that $(a'|b'c')$ is true. It is clear that $\ang{b',c'}\in
\beta^2_{R'}$: as $(b,b),(b,c)\in \al$, we have $(b,b'),(b,c')\in
\rel'$. To get a contradiction suppose that $\ang{a',b'}\in
\beta^2_{\rel'}$. Since $(a,a'),(b,b')\in\rel'$ and $\rel'$ is a 
thick mapping, the pairs $(a,b'),(b,a')$ has to belong to $\rel'$  
as well. That is, there are $x,y$ such that
$\ang{a,x},\ang{b,y}\in \al$ and $(x,b'),(y,a')\in \rel$. Now equivalence
relation $\al$ shows that $(y|ax)$ is true (since 
$\ang{a,x},\ang{b,y}\in \al$ and $\ang{b,a}\not\in \al$) and 
equivalence relation $\alpha^1_{\rel}$
shows that $(x|ay)$ is true (since $(a,a'),(y,a')\in \rel$ shows
$\ang{a,y}\in\alpha^1_{\rel}$, $(b,b'),(x,b')\in \rel$ shows
$\ang{b,x}\in \alpha^1_{\rel}$,  and we know that $\ang{a,b}\not\in
\alpha^1_{\rel}$). By Lemma~\ref{lem:triple}, $(y|ax)$ and $(x|ay)$
cannot be both true, a contradiction.
\end{proof}

Let $\Gm$ be a non-crossing decomposable language over $D$.  Let
$\text{minor}(a,b,c)$ be $a$ if $a,b,c$ are all different or all the
same, otherwise let it be the value that appears only once among $a,b,c$.
Similarly, let $\text{major}(a,b,c)$ be $a$ if $a,b,c$ are all
different or all the same, otherwise let it be the value that appears
more than once among $a,b,c$. Because of Lemma~\ref{lem:triple}, the
following two functions are well defined:

\begin{gather*}
m(a,b,c)=
\begin{cases}
b& \text{if $(a|bc)$ is true,}\\
a& \text{if $(b|ac)$ is true,}\\
b& \text{if $(c|ab)$ is true,}\\
\text{major}(a,b,c) & \text{if none of $(a|bc), (b|ac), (c|ab)$ is true,}\\
\end{cases}\\
h(a,b,c)=
\begin{cases}
a& \text{if $(a|bc)$ is true,}\\
b& \text{if $(b|ac)$ is true,}\\
c& \text{if $(c|ab)$ is true,}\\
\text{minor}(a,b,c) & \text{if none of $(a|bc), (b|ac), (c|ab)$ is true,}\\
\end{cases}
\end{gather*}

\begin{lemma}\label{lem:maltsev}
Operations $m$ and $h$ are conservative majority and minority
operations, respectively, and $\dang\Gm$ is invariant under $m$ and $h$.
\end{lemma}

\begin{proof}
It is clear that $m$ and $h$ are conservative. To show that $m$ is a
majority operation, by definition of major, it is sufficient to consider the case when 
one of $(a|bc), (b|ac), (c|ab)$ is true. If $b=c$ (resp., $a=c$, $a=b$), then
only $(a|bc)$ (resp., $(b|ac)$, $(c|ab)$) can be true, which means
that $m(a,b,c)=b$ (resp., $a$, $b$), as required. A similar argument
shows that $h$ is a minority function.

Since every relation in $\dang\Gm$ is 2-decomposable, it is sufficient
to show that every binary relation in $\dang\Gm$ is invariant under
$m$ and $h$. We show invariance under $h$, the proof is similar for $m$.
Let $R\in \dang\Gm$ be a binary relation, which is a thick mapping by
assumption. 
 Take
$(a,a'),(b,b'),(c,c')\in\rel$. If $a,b,c$ are in the same $\al_\rel^1$-class
then $a',b',c'$ are in the same $\al_\rel^2$-class. Since
$h(a,b,c)\in\{a,b,c\}$ and $h(a',b',c')\in\{a',b',c'\}$, it follows
that $(h(a,b,c),h(a',b',c'))\in\rel$. If $a,b,c$ are not all in the
same $\al_\rel^1$-class, and one of  $(a|bc), (b|ac), (c|ab)$ is
true, then by Lemma~\ref{lem:triple2}, the corresponding statement from
$(a'|b'c'), (b'|a'c'), (c'|a'b')$ is also true. Now the pair
$(h(a,b,c),h(a',b',c'))$ has to be one of $(a,a')$, $(b,b')$,
$(c,c')$, hence it is in $R$. If none of $(a|bc), (b|ac), (c|ab)$ is
true, then $a,b,c$ are in different $\al^1_\rel$-classes. Moreover,  if
none of $(a'|b'c')$, $(b'|a'c'), (c'|a'b')$ is
true, then $a',b',c'$ are in different $\al^2_\rel$-classes. Therefore
$(h(a,b,c),h(a',b',c'))=(\text{minor}(a,b,c),\text{minor}(a',b',c'))=
(a,a')\in\rel$.
\end{proof}

\begin{remark}\label{rem:minority}
Interestingly, Lemma~\ref{lem:maltsev} (along with Lemma~\ref{lem:equiv0}) 
gives more than just the existence of a majority polymorphism, and 
a conservative Mal'tsev polymorphism. The operations $m$ and 
$h$ are both conservative, and $h$ is a minority operation, not just 
a Mal'tsev one. Therefore, we have that a constraint language has 
a majority and conservative Mal'tsev polymorphisms if and only if it 
has a majority and minority polymorphisms, both conservative.
\end{remark}

The following two consequences of having a conservative Mal'tsev
polymorphism will be used in the algorithm.
\begin{lemma}\label{lem:transitive}
Let $\Gm$ be a constraint language having a conservative
Mal'tsev polymorphism. Let $\rel,\rel'\in \dang\Gm$ be two nontrivial thick
mappings such
that $\pr_2\rel=\pr_1\rel'$. Then $\rel\circ\rel'$ is also
non-trivial. 
\end{lemma}
\begin{proof}
Since $\rel$, $\rel'$ are nontrivial, there are values $a$, $a'$, $b$,
$b'$ such that $\ang{a,a'}\not\in \alpha^1_{\rel}$, $\ang{b,b'}\not\in
\alpha^2_{\rel'}$. If $\rel\circ\rel'$ is trivial, then
$(a,b'),(a,b),(a',b)\in \rel\circ\rel'$, which means that there are
(not necessarily distinct) values $x,y,z$ such that
\begin{gather*}
(a,x)\in \rel,\ (x,b')\in \rel',\\
(a,y)\in \rel,\ (y,b)\in \rel',\\
(a',z)\in \rel,\ (z,b)\in \rel'.\\
\end{gather*}
Let $m$ be a conservative Mal'tsev polymorphism. Let
$q=m(x,y,z)\in \{x,y,z\}$. Applying $m$ on the pairs above, we get
that $(a',q)\in \rel$ and $(q,b')\in \rel'$. It is not possible that
$q\in \{x,y\}$, since this would mean that
$\ang{a,a'}\in\alpha^1_{\rel}$. It is not possible that $q=z$ either, since
that would imply $\ang{b,b'}\in
\alpha^2_{\rel'}$, a contradiction.
\end{proof}

Recall that if $\alpha$ and $\beta$ are equivalence relations on the
same set $S$, then $\alpha \join\beta$ denotes the smallest (in
terms of the number of pairs it contains) equivalence relation on $S$
containing both $\alpha$ and $\beta$. It is well known that if $\al$
and $\beta$ have a Mal'tsev polymorphism then they \emph{permute},
that is, $\al\circ\beta=\beta\circ\al=\al\vee\beta$.

\begin{lemma}\label{lem:non-trivial}
Let $\Gm$ be a finite constraint language having a conservative
Mal'tsev polymorphism. Let $\alpha,\beta\in \dang\Gm$ be two
nontrivial equivalence relations on the same set $S$. Then
$\alpha\join\beta$ is also non-trivial. 
\end{lemma}

\begin{proof}
Since $\alpha \join
\beta=\alpha\circ\beta$, by Lemma~\ref{lem:transitive}, it is non-trivial.
\end{proof}

We will also need the following observation.

\begin{lemma}\label{lem:join-union}
Let $\Gm$ be a finite constraint language having a majority 
polymorphism and a conservative
Mal'tsev polymorphism. Let $\alpha,\beta\in \dang\Gm$ be two
equivalence relations on the same set $S$. Then
$\alpha\join\beta=\al\cup\beta$. 
\end{lemma}

\begin{proof}
By Lemma~\ref{lem:equiv0} $\al$ and $\beta$ are non-crossing. 
Then if $\ang{a,b}\in\al$ and $\ang{b,c}\in\beta$ then $\ang{b,c},
\ang{a,c}\in\al$. The result follows.
\end{proof}

%%%%%%%%%%%%%%%%%%%%%%%%%%%%%%%%%%%%%%%%%%%%%%%%%%%%%% 
%%%%%%%%%%%%%%%%%%%%%%%%%%%%%%%%%%%%%%%%%%%%%%%%%%%%%%
\section{Algorithm}\label{sec:algorithm}

In this section we fix a constraint language $\Gm$ that has
conservative majority and minority polymorphisms.
We present a
polynomial-time algorithm for solving $\CCSP(\Gm)$ and $\NCCSP(\Gm)$ in this case.

%%%%%%%%%%%%%%%%%%%%%%%%%%%%%%%%%%%%%%%%%%%%%%%%%%%%%%
\subsection{Prerequisites}

In this subsection we prove several properties of instances of $\CCSP(\Gm)$ and 
$\NCCSP(\Gm)$ that will be very instrumental for our algorithms. First of all, we 
show that every such instance can be supposed binary, that is, that every its 
constraint is imposed only on two variables. Then we introduce a graph 
corresponding to such an instance, and show that if this graph is disconnected 
then a solution to the whole problem can be obtained by combining arbitrarily 
solutions for the connected components. Finally, if the graph is connected, the 
set of possible values for each variable can be subdivided into several subsets, 
so that if the variable takes a value from one of the subsets, then each of the 
remaining variables is forced to take values from a particular subset of the 
corresponding partition. 

Observe that if a constraint language $\Gm$ satisfies the conditions of Theorem~\ref{the:main-poly} then by Remark~\ref{rem:minority} the constraint language $\Gm'$ obtained from $\Gm$ by adding all unary relations also satisfies the conditions of Theorem~\ref{the:main-poly}. Indeed, $\Gm$ has conservative majority and minority polymorphisms that are also polymorphisms of $\Gm'$. Therefore we will assume that $\Gm$ contains all unary relations. It will also be convenient to assume that $\Gm$ contains all the binary relations from $\dang\Gm$.

Let $\Gm$ be a constraint language and let $\cP=(V,\Cc)$ be a
2-consistent instance of $\CSP(\Gm)$. By $\bin(\cP)$ we denote the
instance $(V,\Cc')$ such that $\Cc'$ is the set of all constraints of
the form $\ang{(v,w),\rel_{v,w}}$ where $v,w\in V$ and $\rel_{v,w}$ is
the set of all partial solutions on $\{v,w\}$. 

\begin{lemma}\label{lem:binary-solutions}
Let $\Gm$ be a constraint language with a majority polymorphism. Then
if $\cP$ is a 2-consistent instance 
of $\CSP(\Gm)$ then $\bin(\cP)$ has the same solutions as $\cP$.  
\end{lemma}

\begin{proof}
Let us denote by $\rel,\rel'$ the $|V|$-ary relations
consisting of all solutions of $\cP$ and $\bin(\cP)$, respectively.
Relations $\rel,\rel'$ are pp-definable in $\Gm$ without equalities,
and $\rel\sse\rel'$. To show that $\rel=\rel'$ we use the result from
\cite{Jeavons98:consist} stating that, since 
$\Gm$ has a majority polymorphism for any
$v,w\in V$ and any $(a,b)\in\rel_{v,w}$ we have
$(a,b)\in\pr_{v,w}\rel$, i.e.\ $\pr_{v,w}\rel=\pr_{v,w}\rel'$. Since
$\rel$ is 2-decomposable, if $\ba\in\rel'$, that is,
$\pr_{v,w}\ba\in\rel_{v,w}=\pr_{v,w}\rel$ for all $v,w\in V$, then
$\ba\in\rel$.
\end{proof}

Let $\cP=(V,\Cc)$ be an instance of $\CSP(\Gm)$. Applying algorithm
{\sc 2-Consistency} we may assume that $\cP$ is 2-consistent. 
By the assumption about $\Gm$,
all constraint relations from $\cP$ are 
2-decomposable, and $\bin(\cP)$ has the same solutions as $\cP$ itself. Therefore, replacing $\cP$ with $\bin(\cP)$, if necessary, every constraint of 
$\cP$ can be assumed to be binary.

Let constraints of $\cP$ be $\ang{(v,w),\rel_{vw}}$ for each
pair of different $v,w\in V$. Let $\cS_v$, $v\in V$, denote the set of
$a\in D$ such that there is a solution $\vf$ of $\cP$ such that
$\vf(v)=a$. By \cite{Jeavons98:consist}, $\cP$ is globally
consistent, therefore, $\cS_v=\pr_1\rel_{vw}$ 
for any $w\in V$, $w\ne v$. 

Constraint $C_{vw}=\ang{(v,w),\rel_{vw}}$ is said
to be \emph{trivial} if $\rel_{vw}=\cS_v\tm\cS_w$, otherwise it is
said to be \emph{non-trivial}. 
The \emph{graph of} $\cP$, denoted $G(\cP)$, is a graph with vertex
set $V$ and edge set $E=\{vw\mid v,w\in V \text{ and $C_{vw}$
is non-trivial}\}$.  

The 2-consistency of $\cP$ implies, in particular, the following simple property.

\begin{lemma}\label{lem:2-consistency}
By the 2-consistency of $\cP$, for any $u,v,w\in V$,
$\rel_{uv}\sse\rel_{uw}\circ\rel_{wv}$. 
\end{lemma}

Therefore, by Lemma~\ref{lem:transitive}, the graph $G(\cP)$ is
transitive, i.e.,  every connected component is a clique.

If $G(\cP)$ is not connected, every combination of solutions for its 
connected components give rise to a solution of the entire problem. 
More precisely, let $V_1\zd V_k$ be the connected components 
of $G(\cP)$, and let $\cP_{|V_i}$ denote the instance $(V_i,\Cc_i)$ 
where $\Cc_i$ includes all the constraints $\ang{(v,w),\rel_{vw}}$ 
for which $v,w\in V_i$. We will use the following observation.

\begin{lemma}\label{lem:disconnected}
Let $\vf_1\zd\vf_k$ be solutions of $\cP_{|V_1}\zd\cP_{|V_k}$. Then the mapping $\vf: V\to D$ such that $\vf(v)=\vf_i(v)$ whenever $v\in V_i$ is a solution of $\cP$.
\end{lemma}

\begin{proof}
  We need to check that all constraints of $\cP$ are satisfied.
  Consider $C_{vw}=\ang{(v,w),\rel_{vw}}$. If $v,w\in V_i$ for a
  certain $i$, then $(\vf(v),\vf(w))=(\vf_i(v),\vf_i(w))\in\rel_{vw}$
  since $\vf_i$ is a solution to $\cP_i$. If $v,w$ belong to different
  connected components, then $\rel_{vw}$ is trivial, and so $C_{vw}$ is
  satisfied.
\end{proof}

Suppose that $G(\cP)$ is connected and fix $v\in V$. In this case, the
graph is a clique, and therefore for any $w\in V$ the constraint
$C_{vw}$ is non-trivial. Note that due to 2-consistency, every
relation $\al_{\rel_{vw}}^1$ for $w\in V\setminus\{v\}$ is over the
same set $\cS_v$. Set $\eta_v=\bigvee_{w\in
  V\setminus\{v\}}\al_{\rel_{vw}}^1$; as every $\al_{\rel_{vw}}$ is
non-trivial, Lemma~\ref{lem:non-trivial} implies that $\eta_v$ is
non-trivial.

\begin{lemma}\label{lem:variable-non-trivial}
Suppose $G(\cP)$ is connected. Equivalence relations $\eta_v$ and $\al_{\rel_{vw}}^1$ (for any $w\in
V\setminus\{v\}$) are non-trivial.
\end{lemma}

\begin{lemma}\label{lem:partition}
Suppose $G(\cP)$ is connected.\\[1mm]
(1) For any $v,w\in V$ there is a one-to-one correspondence
$\psi_{vw}$ between $\cS_v\fac{\eta_v}$ and $\cS_w\fac{\eta_w}$ such
that for any solution $\vf$ of $\cP$ if $\vf(v)\in
A\in\cS_v\fac{\eta_v}$, then $\vf(w)\in \psi_{vw}(A)\in\cS_w\fac{\eta_w}$.\\[1mm] 
(2) The mappings $\psi_{vw}$ are consistent, i.e.\ for any $u,v,w\in
V$ we have $\psi_{uw}(x)=\psi_{vw}(\psi_{uv}(x))$ for every $x$.
\end{lemma}

\begin{proof}
(1) Let $\rel_{vw}$ be a thick mapping with respect to 
    a mapping $\vr$, and $\al=\al_{\rel_{vw}}^1$, $\al'=\al_{\rel_{vw}}^2$. 
    Recall that $\vr$ is a one-to-one mapping from
    $\cS_v\fac\al$ to $\cS_w\fac{\al'}$. Suppose, to obtain a contradiction, 
    that $\vr$ does not
    induce a one-to-one mapping between $\cS_v\fac{\eta_v}$ and
    $\cS_w\fac{\eta_w}$. Then without loss of generality there are
    $a,b\in\cS_v$ such that $\ang{a,b}\in\eta_v$, but for certain
    $a',b'\in\cS_w$ we have $(a,a'),(b,b')\in\rel_{vw}$ and
    $\ang{a',b'}\not\in\eta_w$. Since $\al'\sse\eta_w$,
    $\ang{a',b'}\not\in\al'$, hence $\ang{a,b}\not\in\al$. By 
    Lemma~\ref{lem:join-union} there is
    $u\in V$ such that $\rel_{vu}$ is a thick mapping with respect to
    some $\beta,\beta'$ and $\ang{a,b}\in\beta$. Therefore for some
    $c\in\cS_u$ we have $(a,c),(b,c)\in\rel_{vu}$. Since
    $\rel_{vu}\sse\rel_{vw}\circ\rel_{wu}$, there exist
    $d_1,d_2\in\cS_w$ satisfying the conditions
    $(a,d_1),(b,d_2)\in\rel_{vw}$ and
    $(d_1,c),(d_2,c)\in\rel_{wu}$. The first pair of inclusions imply
    that $\ang{a',d_1},\ang{b',d_2}\in\al'$, while the second one
    implies that $\ang{d_1,d_2}\in\eta_w$. Since $\al'\sse\eta_w$, we
    obtain $\ang{a',b'}\in\eta_w$, a contradiction.  

(2) If for some $u,v,w\in V$ there is a class $A\in\cS_u\fac{\eta_u}$
    such that $\psi_{vw}(\psi_{uv}(A))\ne\psi_{uw}(A)$ then
    $\rel_{uw}\not\sse\rel_{uv}\circ\rel_{vw}$, a contradiction. 
\end{proof}

Fix a variable $v_0$ of $\cP$ and take a $\eta_{v_0}$-class $A$. Let 
$\cP_A=(V,\Cc_A)$ denote the problem instance over the same 
variables, where for every
$v,w\in V$ the set $\Cc_A$ includes the constraint $\ang{(v,w),\rel^A_{vw}}$ with
$\rel^A_{vw}=\rel_{vw}\cap(\psi_{v_0v}(A)\tm\psi_{v_0w}(A))$.

\begin{lemma}\label{lem:unary restrictions}
Problem $\cP_A$ belongs to $\CCSP(\Gm)$.
\end{lemma}

\begin{proof}
It suffices to show that $\rel^A_{vw}\in\Gm$ for any $v,w\in V$. By 
Lemma~\ref{lem:consistency} $\rel_{vw}\in\dang\Gm$, and as we 
assumed that  $\Gm$ contains all binary relations from $\dang\Gm$, 
we have $\rel_{vw}\in\Gm$. By the assumption made, all 
unary relations including $\psi_{v_0v}(A)$ and $\psi_{v_0w}(A)$ 
belong to $\Gm$. Therefore relation $\rel^A_{vw}$ is pp-definable 
in $\Gm$, and, as a binary relation, belongs to it.
\end{proof}

%%%%%%%%%%%%%%%%%%%%%%%%%%%%%%%%%%%%%%%%%%%%%%%%%%%%%%%
\subsection{Algorithm: The decision problem}

We split the algorithm into two parts. Algorithm {\sc Cardinality}
(Figure~\ref{fig:alg-cardinality}) just ensures 2-consistency and initializes
a recursive process. The main part of the work is done by
{\sc Ext-Cardinality} (Figure~\ref{fig:alg-vector}).

Algorithm {\sc Ext-Cardinality} solves the more general problem of
computing the set of all cardinality constraints $\pi$ that can be
satisfied by a solution of $\cP$. Thus it can be used to solve
directly CSP with extended global cardinality constraints, see Preliminaries.

The algorithm considers three
cases. Step 2 handles the trivial case when the instance consists of a
single variable and there is only one possible value it can be
assigned. Otherwise, we decompose the instance either by partitioning
the variables or by partitioning the domain of the variables. If
$G(\cP)$ is not connected, then the satisfying assignments of $\cP$
can be obtained from the satisfying assignments of the connected
components. Thus a cardinality constraint $\pi$ can be satisfied if it
arises as the sum $\pi_1+\dots+\pi_k$ of cardinality constraints such
that the $i$-th component has a solution satisfying $\pi_i$. Instead
of considering all such sums (which would not be possible in
polynomial time), we follow the standard dynamic programming approach
of going through the components one by one, and determining all
possible cardinality constraints that can be satisfied by a solution
for the first $i$ components (Step 3).

If the graph $G(\cP)$ is connected, then we fix a variable $v_0$ and
go through each class $A$ of the partition $\eta_{v_0}$ (Step 4). If $v_0$ is
restricted to $A$, then this implies a restriction for every other
variable $w$. We recursively  solve the problem for the restricted
instance $\cP_A$ arising for each class $A$; if  constraint $\pi$
can be satisfied, then it can be satisfied for one of
the restricted instances.

The correctness of the algorithm follows from the discussion
above. The only point that has to be verified is that the instance
remains 2-consistent after the recursion. This is obvious if we
recurse on the connected components (Step 3). In Step 4, 2-consistency follows
from the fact that if $(a,b)\in R_{vw}$ can be extended by $c\in S_u$,
then in every subproblem either these three values satisfy the
instance restricted to $\{v,w,u\}$ or $a$, $b$, $c$ do not appear in
the domain of $v$, $w$, $u$, respectively. 

To show that the algorithm runs in polynomial time, observe first that
every step of the algorithm (except the recursive calls) can be done
in polynomial time. Here we use that $D$ is fixed, hence the size of
the set $\Pi$ is polynomially bounded. Thus we only need to bound the size of the
recursion tree. If we recurse in Step~3, then we produce instances
whose graphs are connected, thus it cannot be followed by recursing
again in Step~3. In Step~4, the domain of every variable is decreased:
by Lemma~\ref{lem:variable-non-trivial}, $\eta_w$ is nontrivial for
any variable $w$. Thus in any branch of the recursion tree, recursion
in Step 4 can occur at most
$|D|$ times, hence the depth of the recursion tree is $O(|D|)$. As the
number of branches is polynomial in each step, 
the size of the recursion tree is polynomial.
{\small
\begin{figure}[ht]
\begin{tabbing}
INPUT: \ \ \ \ \ \ \= An instance $\cP=(V,\Cc)$ of $\CCSP(\Gm)$ with
a cardinality constraint $\pi$\\ 
OUTPUT: \> YES if $\cP$ has a solution satisfying $\pi$, NO otherwise\\[2mm]
{\em Step 1.} \> {\bf apply} {\sc 2-Consistency} to $\cP$\\
{\em Step 2.} \> {\bf set} $\Pi:=${\sc Ext-Cardinality}$(\cP)$\\
{\em Step 3.} \> {\bf if} $\pi\in\Pi$ {\bf output} YES\\
\> {\bf else output} NO
\end{tabbing}
\caption{Algorithm {\sc Cardinality}.}\label{fig:alg-cardinality}
\end{figure}

\begin{figure}[ht]
\begin{tabbing}
INPUT: \ \ \ \ \ \ \= A 2-consistent instance $\cP=(V,\Cc)$ of $\CCSP(\Gm)$\\
OUTPUT: \> The set of cardinality constraints $\pi$ such that $\cP$
has a solution that satisfies $\pi$\\[2mm] 
{\em Step 1.} \> {\bf construct} the graph $G(\cP)=(V,E)$\\
{\em Step 2.} \> {\bf if} $|V|=1$ and the domain of this variable is
a singleton $\{a\}$ {\bf then do}\\
{\em Step 2.1} \> \ \ \ \ \= {\bf set} $\Pi:=\{\pi\}$ where $\pi(x)=0$
except for $\pi(a)=1$\\
{\em Step 3.} \> {\bf else if} $G(\cP)$ is disconnected and 
$G_1=(V_1,E_1)\zd G_k=(V_k,E_k)$ are\\
\> its connected components {\bf then do}\\ 
{\em Step 3.1} \> \> {\bf set} $\Pi:=\{\pi\}$ where $\pi:D\to\nat$ is given by $\pi(a)=0$ for $a\in D$\\
{\em Step 3.2} \> \> {\bf for} $i=1$ {\bf to} $k$ {\bf do}\\
{\em Step 3.2.1} \> \> \ \ \ \ \= {\bf set} $\Pi:=\Pi\;+\;${\sc
  Ext-Cardinality}$(\cP_{|V_i})$\\ 
\> \> {\bf endfor}\\
\> {\bf endif}\\
{\em Step 4.} \> {\bf else do}\\
{\em Step 4.1} \> \> {\bf for each} $v\in V$ {\bf find} $\eta_v$\\
{\em Step 4.2} \> \> {\bf fix} $v_0\in V$ {\bf and set} $\Pi:=\emptyset$\\
{\em Step 4.3} \> \> {\bf for each} $\eta_{v_0}$-class $A$ {\bf do}\\
{\em Step 4.3.1} \> \> \> {\bf set} $\cP_A:=(V,\Cc_A)$ where for every
$v,w\in V$ the set $\Cc_A$ includes\\ 
\> \> \> the constraint $\ang{(v,w),\rel_{vw}\cap(\psi_{v_0v}(A)\tm\psi_{v_0w}(A))}$\\
{\em Step 4.3.2} \> \> \> {\bf set} $\Pi:=\Pi\; \cup \;${\sc Ext-Cardinality}$(\cP_A)$\\
\> \> {\bf endfor}\\
\> {\bf enddo}\\
{\em Step 4.} \> {\bf output} $\Pi$
\end{tabbing}
\caption{Algorithm {\sc Ext-Cardinality}.}\label{fig:alg-vector}
\end{figure}
}
%%%%%%%%%%%%%%%%%%%%%%%%%%%%%%%%%%%%%%%%%%%%%%%%%%%%%
\subsection{Solving the counting problem}

In this section we observe that algorithm {\sc Cardinality} can be
modified so that it also solves counting CSPs with global cardinality constraints,
provided $\Gm$ satisfies the conditions of
Theorem~\ref{the:main-poly}. 

The counting algorithm works
very similar to algorithm {\sc Cardinality}, except that instead of
determining the set of satisfiable cardinality constraints, it keeps track of the
number of solutions that satisfy every cardinality constraint
possible. It considers the same 3 cases. In the trivial case of a
problem with one variable and one possible value for this variable,
the algorithm assigns 1 to the cardinality constraint satisfied by the
only solution of the problem and 0 to all other cardinality
constraints. In the case of disconnected graph $G(\cP)$ if a
cardinality constraint can be represented in the form
$\pi=\pi_1+\ldots+\pi_k$, then solutions on the connected components of
$G(\cP)$ satisfying $\vc \pi k$, respectively, contribute the product
of their numbers into the number of solutions satisfied by $\pi$. We again use the dynamic programming approach, and, for each $i$ compute the number of solutions on $V_1\cup\ldots\cup V_i$ satisfying every possible cardinality constraint. Observe, that the set of cardinality constraints considered is also changed dynamically, as the number of variables grows.
Finally, if $G(\cP)$ is connected, then the different restrictions
have disjoint sets of solutions, hence the numbers of solutions are computed independently.
{\small
\begin{figure}[ht]
\begin{tabbing}
INPUT: \ \ \ \ \ \ \= An instance $\cP=(V,\Cc)$ of $\NCCSP(\Gm)$ with a cardinality constraint $\pi$\\ 
OUTPUT: \> The number of solutions of $\cP$ that satisfy $\pi$\\[2mm]
{\em Step 1.} \> {\bf apply} {\sc 2-Consistency} to $\cP$\\
{\em Step 2.} \> {\bf set} $\vr:=${\sc \#Ext-Cardinality}$(\cP)$\\
\> \ \ \ \ \ \% $\vr(\pi')$ is the number of solutions of $\cP$ satisfying cardinality constraint $\pi'$\\
{\em Step 3.} \> {\bf output} $\vr(\pi)$\\
\end{tabbing}
\caption{Algorithm {\sc \#Cardinality}.}\label{fig:count-card}
\end{figure}

\begin{figure}[ht]
\begin{tabbing}
INPUT: \ \ \ \ \ \ \= A 2-consistent instance $\cP=(V,\Cc)$ of $\NCCSP(\Gm)$\\ 
OUTPUT: \> Function $\vr$ that assigns to every cardinality constraint $\pi$ with $\sum_{a\in D}\pi(a)=|V|$,\\
\> the number $\vr(\pi)$ of solutions of $\cP$ that satisfy $\pi$\\[2mm] 
{\em Step 1.} \> {\bf construct} the graph $G(\cP)=(V,E)$\\
{\em Step 2.} \> {\bf if} $|V|=1$ and the domain of this variable is a singleton $\{a\}$ {\bf then do}\\ 
{\em Step 2.1} \> \ \ \ \ \= {\bf set} $\vr(\pi):=1$ where $\pi(x)=0$ except 
for $\pi(a)=1$, and $\vr(\pi'):=0$ for all $\pi'\ne\pi$\\
\> \> with $\sum_{x\in D}\pi'(x)=1$\\ 
{\em Step 3.} \> {\bf else if} $G(\cP)$ is disconnected and 
$G_1=(V_1,E_1)\zd G_k=(V_k,E_k)$ are its\\
\> connected components {\bf then do}\\ 
{\em Step 3.1} \> \> {\bf set}
$\Pi:=\{\pi\}$ where $\pi:D\to\nat$ is given by $\pi(a)=0$ for $a\in D$, $\vr(\pi):=1$ for $\pi\in\Pi$\\
{\em Step 3.2} \> \> {\bf for} $i=1$ {\bf to} $k$ {\bf do}\\
{\em Step 3.2.1} \> \> \ \ \ \ \= {\bf set} $\Pi':=\{\pi:D\to\nat\mid
\sum_{a\in D}\pi(a)=|V_i|\}$ and $\vr':=${\sc
  \#Ext-Cardinality}$(\cP_{|V_i})$\\ 
{\em Step 3.2.2} \> \> \> {\bf set} $\Pi'':=\{\pi:D\to\nat\mid
\sum_{a\in D}\pi(a)=|V_1|+\ldots+|V_i|\}$, $\vr''(\pi):=0$ for $\pi\in\Pi''$\\
{\em Step 3.2.3} \> \> \> {\bf for each} $\pi\in\Pi$ {\bf and}
$\pi'\in\Pi'$ {\bf set} $\vr''(\pi+\pi'):=\vr''(\pi+\pi')+\vr(\pi)\cdot\vr'(\pi')$\\ 
{\em Step 3.2.4} \> \> \> {\bf set} $\Pi:=\Pi''$, $\vr:=\vr''$\\
\> \> {\bf endfor}\\
\> {\bf endif}\\
{\em Step 4.} \> {\bf else do}\\
{\em Step 4.1} \> \> {\bf for each} $v\in V$ {\bf find} $\eta_v$\\
{\em Step 4.2} \> \> {\bf fix} $v_0\in V$ {\bf and set} $\vr(\pi):=0$
for $\pi$ with $\sum_{a\in D}\pi(a)=|V|$\\
{\em Step 4.3} \> \> {\bf for each} $\eta_{v_0}$-class $A$ {\bf do}\\
{\em Step 4.3.1} \> \> \> {\bf set} $\cP_A:=(V,\Cc_A)$ where for every 
$v,w\in V$ the set $\Cc_A$ includes the constraint\\ 
\> \> \> $\ang{(v,w),\rel_{vw}\cap(\psi_{v_0v}(A)\tm\psi_{v_0w}(A))}$\\
{\em Step 4.3.2} \> \> \> {\bf set} $\vr':=${\sc \#Ext-Cardinality}$(\cP_A)$\\
{\em Step 4.3.3} \> \> \> {\bf set} $\vr(\pi):=\vr(\pi)+\vr'(\pi)$\\
\> \> {\bf endfor}\\
\> {\bf enddo}\\
{\em Step 4.} \> {\bf output} $\vr$
\end{tabbing}
\caption{Algorithm {\sc \#Ext-Cardinality}.}\label{fig:count-ext-card}
\end{figure}
}
%%%%%%%%%%%%%%%%%%%%%%%%%%%%%%%%%%%%%%%%%%%%%%%%%%%%%
%%%%%%%%%%%%%%%%%%%%%%%%%%%%%%%%%%%%%%%%%%%%%%%%%%%%%
\section{Definable relations, constant relations, and the complexity
  of CCSP}
\label{sec:defin-relat-const}
We present two reductions that will be crucial for
the proofs in Section~\ref{sec:hardness}. In
Section~\ref{sec:pp-definitions}, we show that adding relations
that are pp-definable (without equalities) does not make the problem
harder, while in Section~\ref{sec:constants}, we show the
same for unary constant relations.

%%%%%%%%%%%%%%%%%%%%%%%%%%%%%%%%%%%%%%%%%%%%%%%%%%%%%%%%%%%%%%%%%%%%%%%%%%%%
\subsection{Definable relations and the complexity of cardinality
  constraints}\label{sec:pp-definitions} 

\begin{theorem}\label{the:decision-pp-defin}
Let $\Gm$ be a constraint language and $\rel$ a relation
pp-definable in $\Gm$ without equalities. Then
$\CCSP(\Gm\cup\{\rel\})$ is polynomial time reducible to $\CCSP(\Gm)$. 
\end{theorem}

\begin{proof}
We proceed by induction on the structure of pp-formulas. The base case
of induction is given by $\rel\in\Gm$. We need to consider two cases. 

\medskip

{\sc Case 1.} $\rel(\vc xn)=\rel_1(\vc xn)\wedge\rel_2(\vc xn)$.

\smallskip

Observe that by introducing `fictitious' variables for predicates
$\rel_1,\rel_2$ we may assume that both relations involved have the
same arity. A reduction from $\CCSP(\Gm\cup\{\rel\})$ to $\CCSP(\Gm)$
is trivial: in a given instance of the first problem replace each
constraint of the form $\ang{(\vc vn),\rel}$ with two constraints
$\ang{(\vc vn),\rel_1}$ and $\ang{(\vc vn),\rel_2}$. 

\medskip

{\sc Case 2.} $\rel(\vc xn)=\exists x \rel'(\vc xn, x)$.

\medskip

Let $\cP=(V,\Cc)$ be a $\CCSP(\Gm\cup\{\rel\})$ instance. Without loss
of generality let $\vc Cq$ be the constraints that involve
$\rel$. Instance $\cP'$ of $\CCSP(\Gm)$ is constructed as follows. 
\begin{enumerate}[$\bullet$]
\item
Variables: Replace every variable $v$ from $V$ with a set $W_v$ of variables of
size $q|D|$ and introduce a set of $|D|$ variables for each constraint
involving $R$. More formally, 
$$
W=\bigcup_{v\in V} W_v\cup\{\vc wq\}\cup \bigcup_{i=1}^q \{w^1_i\zd w^{|D|-1}_i\}.
$$
\item
Non-$\rel$ constraints: For every $C_i=\ang{(\vc v\ell),\relo}$ with
$i>q$, introduce all possible constraints of the form $\ang{(\vc
  u\ell),\relo}$, where $u_j\in W_{v_j}$ for $j\in\{1\zd\ell\}$. 
\item
$\rel$ constraints: For every $C_i=\ang{(\vc v\ell),\rel}$, $i\le q$,
  introduce all possible constraints of the form $\ang{(\vc
    u\ell,w_i),\rel'}$, where $u_j\in W_{v_j}$, $j\in\{1\zd \ell\}$. 
\end{enumerate}

{\sc Claim 1.} If $\cP$ has a solution satisfying cardinality
constraint $\pi$ then $\cP'$ has a solution satisfying the cardinality
constraint $\pi'=q|D|\cdot\pi+q$. 

\smallskip 

Let $\vf$ be a solution of $\cP$ satisfying $\pi$. It is
straightforward to verify that the following mapping $\psi$ is a
solution of $\cP'$ and satisfies $\pi'$: 
\begin{enumerate}[$\bullet$]
\item
for each $v\in V$ and each $u\in W_v$ set $\psi(u)=\vf(v)$;
\item
for each $w_i$, where $C_i=\ang{(\vc vn),\rel}$, set $\psi(w_i)$ to be
a value such that\linebreak  $(\vf(v_1)\zd\vf(v_n),\psi(w_i))\in\rel'$. 
\item
for each $i\le q$ and $j\le|D|-1$ set $\psi(w_i^j)$ to be such that
$\{\psi(w_i),\psi(w_i^1)\zd \psi(w_i^{|D|-1|})\}=D$. 
\end{enumerate}

{\sc Claim 2.} If $\cP'$ has a solution $\psi$ satisfying the
cardinality constraint $\pi'=q|D|\cdot\pi+q$, then $\cP$ has a
solution satisfying constraint $\pi$.

\smallskip

Let $a\in D$ and $U_a(\psi)=\psi^{-1}(a)=\{u\in W\mid
\psi(u)=a\}$. Observe first that if $\vf:V\to D$ is a mapping such that
$U_{\vf(v)}(\psi)\cap W_v\neq\emptyset$ for every $v\in V$ (i.e.,
$\psi(v')=\vf(v)$ for at least one variable $v'\in W_v$), then 
$\vf$ satisfies all the constraints of $\cP$. Indeed, consider a constraint $C=\ang{\bs,\relo}$
of $\cP$ where $\relo\neq\rel$. Let $\bs=(v_1,\dots,v_\ell)$. For every $v_i$, there is a
$v'_i\in W_{v_i}$ such that $\vf(v_i)=\psi(v'_i)$. By the way $\cP'$
is defined, it contains a constraint $C'=\ang{\bs',\relo}$ where
$\bs'=(v'_1,\dots,v'_\ell)$. Now the fact that $\psi$ satisfies $C'$
immediately implies that $\vf$ satisfies $C$: 
$(\vf(v_1)\zd\vf(v_\ell))=(\psi(v'_1)\zd\psi(v'_\ell))\in\relo$ .
The argument is similar if $\relo=\rel$.

We show that it is possible to construct such a $\vf$ that also
satisfies the cardinality constraint $\pi$.  Since $|W_v|=q|D|$, for any $a\in D$ with $\pi(a)\ne0$, even
if set $U_a(\psi)$ contains all $q|D|$ variables of the form $w_i$ and
$w_i^j$, it has to intersect at least $\pi(a)$ sets $W_v$ (as
$(\pi(a)-1)q|D|+q|D|<\pi'(a)=\pi(a)\cdot q|D|+q$).  Consider the bipartite
graph $G=(T_1\cup T_2,E)$, where $T_1,T_2$ is a bipartition and
\begin{enumerate}[$\bullet$]
\item
$T_1$ is the set of variables $V$;
\item $T_2$ is constructed from the set $D$ of values by taking
  $\pi(a)$ copies of each value $a\in D$;
\item
edge $(v,a')$, where $a'$ is a copy of $a$ from $T_2$, belongs to $E$ if
and only if $W_v\cap U_a(\psi)\ne\emptyset$. 
\end{enumerate}
Note that $|T_1|=|T_2|$ and a perfect matching $E'\sse E$ corresponds
to a required mapping $\vf$: $\vf(v)=a$ if $(v,a')\in E'$ for some
copy $a'$ or $a$. 

Take any subset $S\sse T_2$, let $S$ contains some copies of $\vc as$.
Then by the observation above, $S$
has at least  $\pi(a_1)+\ldots+\pi(a_s)$ neighbours in $T_1$. Since $S$ contains
at most $\pi(a_i)$ copies of $a_i$,
$$
\pi(a_1)+\ldots+\pi(a_s)\ge |S|.
$$
By Hall's Theorem on perfect matchings in bipartite graphs, $G$ has a
perfect matching, concluding the proof that the required $\vf$ exists.
\end{proof}

%%%%%%%%%%%%%%%%%%%%%%%%%%%%%%%%%%%%%%%%%%%%%%%%%%
\subsection{Constant relations and the complexity of cardinality
  constraints}\label{sec:constants} 

Let $D$ be a set, and let $a\in D$. The \emph{constant relation} $C_a$
is the unary relation that contains only one tuple, $(a)$. If a
constraint language $\Gm$ over $D$ contains all the constant
relations, then they can be used in the corresponding constraint
satisfaction problem to force certain variables to take some fixed
values. The goal of this section is to show that for any constraint
language $\Gm$ the problem $\CCSP(\Gm\cup\{C_a\mid a\in D\})$ is
polynomial time reducible to $\CCSP(\Gm)$. For the ordinary decision
CSP such a reduction exists when $\Gm$ does not have unary
polymorphisms that are not permutations, see
\cite{Bulatov05:classifying}. 

We make use of the
notion of multi-valued morphisms, a generalization of homomorphisms,
that in a different context has appeared in the literature for a while 
(see, e.g.\ \cite{Rosenberg98:hyperstructures}) under the guise hyperoperation.
Let $\rel$ be a (say, $n$-ary) relation on a set $D$, and let $f$ be a
mapping from $D$ to $2^D$, the powerset of $D$. Mapping $f$ is said to
be a \emph{multi-valued morphism} of $\rel$ if for any tuple $(\vc
an)\in\rel$ the set $f(a_1)\tm\ldots\tm f(a_n)$ is a subset of
$\rel$. Mapping $f$ is a multi-valued morphism of a constraint
language $\Gm$ if it is a multi-valued morphism of every relation in
$\Gm$. 
For a multi-valued morphism $f$ and set $A\subseteq D$, we
define $f(A):=\bigcup_{a\in A}f(a)$. The product of two
multi-valued morphisms $f_1$ and $f_2$ is defined by $(f_1\circ
f_2)(a):=f_1(f_2(a))$ for every $a\in D$. We denote by $f^i$ the
$i$-th power of $f$, with the convention that $f^0$ maps $a$ to
$\{a\}$ for every $a\in A$. 

\begin{theorem}\label{the:adding-constants}
Let $\Gm$ be a finite constraint language over a set $D$. Then
$\CCSP(\Gm\cup\{C_a\mid a\in D\})$ is polynomial time reducible to $\CCSP(\Gm)$. 
\end{theorem}

\begin{proof}
Let $D=\{\vc dk\}$ and
$a=d_1$. We show that $\CCSP(\Gm\cup\{C_a\})$ is polynomial time reducible to
$\CCSP(\Gm)$. This clearly implies the result. We make use of the
following multi-valued morphism gadget $\MVM(\Gm,n)$ (i.e.\ a CSP
instance). Observe that it is somewhat similar to the \emph{indicator
  problem} \cite{Jeavons99:expressive}. 
\begin{enumerate}[$\bullet$]
\item
The set of variables is $V(n)=\displaystyle\bigcup_{i=1}^k
V_{d_i}$, where $V_{d_i}$ contains $n^{|D|+1-i}$ elements.
All sets $V_{d_i}$ are assumed
to be disjoint. 
\item
The set of constraints is constructed as follows: For every (say,
$r$-ary) $\rel\in\Gm$ and every $(\vc ar)\in\rel$ we include all
possible constraints of the form $\ang{(\vc vr),\rel}$ where $v_i\in
V_{a_i}$ for $i\in\{1\zd r\}$. 
\end{enumerate}

Now, given an instance $\cP=(V,\Cc)$ of $\CCSP(\Gm\cup\{C_a\})$, we
construct an instance $\cP'=(V',\Cc')$ of $\CCSP(\Gm)$.  
\begin{enumerate}[$\bullet$]
\item
Let $W\sse V$ be the set of variables $v$, on which the constant relation
$C_a$ is imposed, that is, $\Cc$ contains the constraint
$\ang{(v),C_a}$. Set $n=|V|$. The set $V'$ of variables of
$\cP'$ is the disjoint union of the set $V(n)$ of variables of
$\MVM(\Gm,n)$ and $V\setminus W$. 
\item
The set $\Cc'$ of constraints of $\cP'$ consists of three parts:
\begin{enumerate}[$-$]
\item[(a)]
$\Cc'_1$, the constraints of $\MVM(\Gm,n)$;
\item[(b)]
$\Cc'_2$, the constraints of $\cP$ that do not include variables from $W$;
\item[(c)]
$\Cc'_3$, for any constraint $\ang{(\vc vn),\rel}\in\Cc$ whose scope
  contains variables constrained by $C_a$ (without loss of generality
  let $\vc v\ell$ be such variables), $\Cc'_3$ contains all constraints
  of the form $\ang{(\vc w\ell,v_{\ell+1}\zd v_n),\rel}$, where $\vc w\ell\in
  V_a$. 
\end{enumerate}
\end{enumerate}

We show that $\cP$ has a solution satisfying a cardinality constraint
$\pi$ if and only if $\cP'$ has a solution satisfying cardinality
constraint $\pi'$ given by 
$$
\pi'(d_i)=\left\{\begin{array}{ll}
\pi(a)+(|V_a|-|W|), & \text{if $i=1$},\\
\pi(d_i)+|V_{d_i}|, & \text{otherwise}.
\end{array}\right.
$$

Suppose that $\cP$ has a right solution $\vf$. Then a required
solution for $\cP'$ is given by 
$$
\psi(v)=\left\{\begin{array}{ll}
\vf(v), & \text{if $v\in V\setminus W$},\\
d_i, & \text{if $v\in V_{d_i}$}.
\end{array}\right.
$$
It is straightforward that $\psi$ is a solution to $\cP'$ and that it satisfies $\pi'$.

Suppose that $\cP'$ has a solution $\psi$ that satisfies $\pi'$. 
Since $\pi'(a)=\pi(a)+n^{|D|}-|W|\ge n^{|D|}-n>|V'\setminus V_a|$, there is $v\in V_a$ such
that $\psi(v)=a$. Thus the assignment
$$
\vf(v)=\left\{\begin{array}{ll}
\psi(v), & \text{if $v\in V\setminus W$},\\
a & \text{if $v\in W$}
\end{array}\right.
$$
is a satisfying assignment $\cP$, but it might not satisfy $\pi$. 
Our goal is to show that $\cP'$ has a solution $\psi$, where $\vf$
obtained this way satisfies $\pi$. Observe that what we need is that
$\psi$ assigns value $d_i$ to exactly $\pi'(d_i)-|V_{d_i}|$
variables of $V\setminus W$.

\medskip

{\sc Claim 1.} Mapping $f$ taking every $d_i\in D$ to the set
$\{\psi(v)\mid v\in V_{d_i}\}$ is a multi-valued morphism of
$\Gm$. 

\smallskip

Indeed, let $(\vc an)\in \rel$, $\rel$ is an ($n$-ary) relation from
$\Gm$. Then by the construction of $\MVM(\Gm,n)$ the instance
contains all the constraints of the form $\ang{(\vc vn),\rel}$ with
$v_i\in V_{a_i}$, $i\in\{1\zd n\}$. Therefore, 
\begin{eqnarray*}
\lefteqn{\hspace*{-1cm}\{\psi(v_1)\mid v_1\in V_{a_1}\}\tm\ldots\tm\{\psi(v_n)\mid
v_n\in V_{a_n}\}}\\
&=&  f(a_1)\tm\ldots\tm f(a_n)\sse\rel. 
\end{eqnarray*}

\medskip

\medskip

{\sc Claim 2.} Let $f$ be the mapping defined in Claim 1. Then $f^*$
defined by $f^*(b):=f(b)\cup \{b\}$ for every $b\in D$ is also a multi-valued morphism of
$\Gm$. 
\smallskip

We show that for every $d_i\in D$, there is an $n_i\ge 1$ such that
$d_i\in f^{j}(d_i)$ for every $j\ge n_i$. Taking the maximum
$n=\max_{1\le i \le k}n_i$ of
all these integers, we get that $d_i\in f^{n+1}(d_i)$ and
$f(d_i)\subseteq f^{n+1}(d_i)$ (since $d_i\in f^{n}(d_i)$) for every $i$, proving the claim.

The proof is by induction on $i$. If $d_i\in f(d_i)$, then we are done
as we can set $n_i=1$ (note that this is always the case for $i=1$,
since we observed above that $\psi$ has to assign value $d_1$ to a variable
of $V_{d_1}$). So let us suppose that
$d_i\not\in f(d_i)$. Let $D_i=\{d_1,\dots,d_i\}$ and let $g_i:D_i\to
2^{D_i}$ be defined by $g_i(d_j):=f(d_j)\cap D_i$. Observe that
$g_i(d_j)$ is well-defined, i.e., $g_i(d_j)\neq \emptyset$: the set
$V_{d_j}$ contains $n^{|D|+1-j}\ge n^{|D|+1-i}$ variables, while the
number of variables which are assigned by $\psi$ values outside $D_i$ is
$\sum_{\ell=i+1}^k\pi'(d_\ell)\le n+\sum_{\ell=i+1}^k
n^{|D|+1-\ell}<n^{|D|+1-i}$.

Let $T:=\bigcup_{\ell\ge 1}g_i^\ell(d_i)$. We claim that $d_i\in T$.
Suppose that $d_i\not \in T$. By the definition of $T$ and the
assumption $d_i\not\in f(d_i)$, for every $b\in T\cup \{d_i\}$, the
variables in $V_b$ can have values only from $T$ and from $D\setminus
D_i$. The total number of variables in  $V_b$, $b\in T\cup
\{d_i\}$ is $\sum_{b\in T\cup \{d_i\}}n^{|D|+1-b}$, while the total
cardinality constraint of the values from $T\cup(D\setminus D_i)$ is
\begin{multline*}
\sum_{b\in T\cup(D\setminus D_i)}\pi'(b) =
n+\sum_{b\in T}n^{|D|+1-b}+\sum_{\ell=i+1}^k
n^{|D|+1-\ell} \\< \sum_{b\in T}n^{|D|+1-b}+n^{|D|+1-i}
= \sum_{b\in T\cup \{d_i\}}n^{|D|+1-b},
\end{multline*}
a contradiction. Thus $d_i\in T$, that is, there is a value $j<i$ such
that $d_j\in f(d_i)$ and $d_i\in f^s(d_j)$ for some $s\ge 1$. By the
induction hypothesis, $d_j\in f^n(d_j)$ for every $n\ge n_j$, thus we have
that $d_i\in f^n(d_i)$ if $n$ is at least $n_i:=1+n_j+s$.
This concludes the proof of Claim 2.
\medskip

Let $D^+$ (resp., $D^-$) be the set of those values $d_i\in D$ that
$\psi$ assigns to more than (resp., less than) $\pi'(i)-|V_{d_i}|$ variables of
$V\setminus W$. It is clear that if $|D^+|=|D^-|=0$, then $\vf$
obtained from $\psi$ satisfies $\pi$. Furthermore, if $|D^+|=0$, then
$|D^-|=0$ as well. Thus suppose that $D^+\neq \emptyset$ and let
$S:=\bigcup_{b\in D^+, s\ge 1}f^s(b)$.

\medskip

{\sc Claim 3.}
$S\cap D^-\neq \emptyset$.
\smallskip

Suppose $S\cap D^- =\emptyset$. Every $b\in S\subseteq 
D\setminus D^-$ is assigned by $\psi$ to at least
$\pi'(b)-|V_b|$ variables in $V\setminus W$, implying that $\psi$ assigns
every such $b$ to at most $|V_b|$ variables in the gadget $\MVM(\Gm,n)$.
Thus the total number of variables in the gadget having value from $S$
is at most $\sum_{b\in S}|V_b|$; in fact, it is strictly less than
that since $D^+$ is not empty. By the definition of $S$, only values
from $S$ can be assigned to variables in $V_b$ for every $b\in S$. 
However, the total
number of these variables is exactly $\sum_{b\in S}|V_b|$, a
contradiction.
\medskip

By Claim 3, there is a value $d^-\in S\cap D^-$, which means that there is a
$d^+\in D^+$ and a sequence of distinct values $b_0=d^+$,  $b_1$, $\dots$, $b_\ell=d^-$ such that
$b_{i+1}\in f(b_{i})$ for every $0\le i < \ell$. Let $v\in
V\setminus W$ be an
arbitrary variable having value $d^+$. We modify $\psi$ the following
way:
\begin{enumerate}
\item The value of $v$ is changed from $d^+$ to $d^-$.
\item For every $0\le i <\ell$, one variable in $V_{b_i}$ with value
  $b_{i+1}$ is changed to $b_i$. Note that since $b_{i+1}\in f(b_i)$ 
  and $b_{i+1}\ne b_i$ such a variable exists.
\end{enumerate}
Note that these changes do not modify the cardinalities of the values:
the second step increases only the cardinality of $b_0=d^+$ (by one)
and decreases only the cardinality of $b_\ell=d^-$ (by one). We have
to argue that the transformed assignment still satisfies the constraints
of $\cP'$. Since $d^-\in f^\ell(d^+)$, the multi-valued morphism $f^*$
of Claim 2 implies that changing $d^+$ to $d^-$ on a single variable
and not changing anything else also gives a satisfying assignment. To
see that the second step does not violate the constraints, observe
first that constraints of type (b) are not affected and constraints of
type (c) cannot be violated (since variables in $V_{d_1}$ are changed
only to $d_1$, and there is already at least one variable with value $d_1$
in $V_{d_1}$). To show that constraints of type (a) are not affected,
it is sufficient to show that the mapping $f'$ described by the gadget
after the transformation is still a multi-valued morphism. This can be
easily seen as $f'(b)\subseteq f(b)\cup \{b_i\}=f^*(b)$, where $f^*$
is the multi-valued morphism of Claim 2.

Thus the modified assignment is still a solution of $\cP'$ satisfying
$\pi'$. It is not difficult to show that repeating this operation, in a
finite number of steps we reach a solution where $D^+=D^-=\emptyset$,
i.e., every value $b\in D(b)$ appears exactly $\pi'(b)-|V(b)|$ times
on the variables of $V\setminus W$. As observed above, this implies
that $\cP$ has a solution satisfying $\pi$.
\end{proof}

%%%%%%%%%%%%%%%%%%%%%%%%%%%%%%%%%%%%%%%%%%%%%%%%%%%%%%
%%%%%%%%%%%%%%%%%%%%%%%%%%%%%%%%%%%%%%%%%%%%%%%%%%%%%%
\section{Hardness}\label{sec:hardness}

Now we prove that if constraint language
$\Gm$ does not
satisfy the conditions of Theorem~\ref{the:main-poly} then $\CCSP(\Gm)$ is
NP-complete. 
First, we can easily simulate the
restriction to a subset of the domain by setting to 0 the cardinality
constraint on the unwanted values:

\begin{lemma}\label{lem:restriction}
For any constraint language $\Gm$ over a set $D$ and any $D'\sse D$,
the problem $\CCSP(\Gm_{|D'})$ is polynomial time reducible to
$\CCSP(\Gm)$. 
\end{lemma}

\begin{proof}
For an instance $\cP'=(V,\Cc')$ of $\CCSP(\Gm_{|D'})$ with a global cardinality constraint $\pi':D'\to\nat$ we construct an instance $\cP=(V,\Cc)$ of $\CCSP(\Gm)$ such that for each $\ang{\bs,\rel_{|D'}}\in\Cc'$ we include $\ang{\bs,\rel}$ into $\Cc$. The cardinality constraint $\pi'$ is replaced with $\pi:D\to\nat$ such that $\pi(a)=\pi'(a)$ for $a\in D'$, and $\pi(a)=0$ for $a\in D\setminus D'$. It is straightforward that $\cP$ has a solution satisfying $\pi$ if and only if $\cP'$ has a solution satisfying $\pi'$.
\end{proof}

Suppose now that a constraint language $\Gm$ does not satisfy the
conditions of Theorem~\ref{the:main-poly}. Then one of the following
cases takes place: (a) $\dang\Gm$ contains a binary relation that is
not a thick mapping; or (b) $\dang\Gm$ contains a crossing pair of
equivalence relations; or (c) $\dang\Gm$ contains a relation which is not
2-decomposable. Since relations occurring in cases (a), (b) are not
redundant, and relations that may occur in case (c) can be assumed to
be not redundant, by Lemma~\ref{lem:redundancy} $\dang\Gm$ can be
replaced with $\dang\Gm'$. Furthermore, by
Theorem~\ref{the:adding-constants} all the constant relations can be
assumed to belong to $\Gm$. We consider these three cases in turn.
Furthermore, by a minimality argument, we can assume that if $\Gm$ is
over $D$, then for every $S\subset D$, constraint language $\Gm_{|S}$
satisfies the requirements of Theorem~\ref{the:main-poly}.

One of the NP-complete problems we will reduce to $\CCSP(\rel)$ is the
{\sc Bipartite Independent Set} problem (or BIS for short). In this
problem given a connected bipartite graph $G$ with bipartition
$V_1,V_2$ and numbers $k_1,k_2$, the goal is to verify if there exists
an independent set $S$ of $G$ such that $|S\cap V_1|= k_1$ and
$|S\cap V_2|= k_2$. The NP-completeness of the problem follows for
example from \cite{DBLP:journals/jcss/ChenK03}, which shows the
NP-completeness of the complement problem under the name {\sc
  Constrained Minimum Vertex Cover}. It is easy to see that the
problem is hard even for graphs containing no isolated vertices. By
representing the edges of a bipartite graph with the relation
$\rel=\{(a,c),(a,d),(b,d)\}$, we can express the problem of finding a
bipartite independent set. Value $b$ (resp., $a$) represents selected
(resp., unselected) vertices in $V_1$, while value $c$ (resp., $d$)
represents selected (resp., unselected) vertices in $V_2$. With this
interpretation, the only combination that relation $\rel$ excludes is
that two selected vertices are adjacent. By
Lemma~\ref{lem:thick-mapping}, if a binary relation is not a thick
mapping, then it contains something very similar to $\rel$. However,
some of the values $a$, $b$, $c$, and $d$ might coincide and the
relation might contain further tuples such as $(c,d)$. Thus we need a
delicate case analysis to show that the problem is NP-hard for binary
relations that are not thick mappings.

\begin{lemma}\label{lem:non-thick-mapping}
Let $\rel$ be a binary relation which is not a thick mapping. Then
$\CCSP(\{\rel\})$ is NP-complete. 
\end{lemma}

\begin{proof}
  Since $\rel$ is not a thick mapping, there are
  $(a,c),(a,d),(b,d)\in\rel$ such that $(b,c)\not\in\rel$. By
  Lemma~\ref{lem:restriction} the problem $\CCSP(\rel')$, where
  $\rel'=\rel_{|\{a,b,c,d\}}$, is polynomial time reducible to
  $\CCSP(\rel)$. Replacing $\rel$ with $\rel'$ if necessary we can
  assume that $\rel$ is a relation over $D=\{a,b,c,d\}$ (note that
  some of those elements can be equal). We suppose that $\rel$ is a
  `smallest' relation that is not a thick mapping, that is, for any
  $\rel'$ definable in $\{\rel\}$ with $\rel'\subset\rel$, the relation $\rel'$
  is a thick mapping, and for any subset $D'$ of $D$ the restriction
  of $\rel$ onto $D'$ is a thick mapping.

Let $B=\{x\mid (a,x)\in\rel\}$. Since $B(x)=\exists y R(y,x)\meet C_a(y)$, 
the unary relation $B$ is definable in $\rel$. If $B\ne\pr_2\rel$, 
by setting $\rel'(x,y)=\rel(x,y)\wedge B(y)$ we get a binary
  relation $\rel'$ that is not a thick mapping. Thus by the minimality
  of $\rel$, we may assume that $(a,x)\in\rel$ for any
  $x\in\pr_2\rel$, and symmetrically, $(y,d)\in\rel$ for any
  $y\in\pr_1\rel$.

\smallskip

{\sc Case 1.} $|\{a,b,c,d\}|=4$.

\smallskip

We claim that $|\pr_1\rel|=|\pr_2\rel|=2$. Suppose, without loss of
generality that $x\in \{a,b\}$ appears in $\pr_2\rel$. If $(b,x)\in
\rel$, then $(a,x),(b,x),(a,c)\in\rel$ and $(b,c)\not\in\rel$. Therefore 
the restriction $\rel_{|\{a,b,c\}}$ is not a thick mapping,
contradicting the minimality of $\rel$. Otherwise $(a,x),(a,d),(b,d)\in\rel$ 
while $(b,x)\not\in\rel$. Hence $\rel_{|\{a,b,d\}}$ is
not a thick mapping. Thus we have $\pr_1\rel=\{a,b\}$ and 
$\pr_2\rel=\{c,d\}$.

Let $G=(V,E), V_1, V_2, k_1,k_2$ be an instance of BIS. Construct an
instance $\cP=(V,\Cc)$ by including into $\Cc$, for every edge $(v,w)$
of $G$, the constraint $\ang{(v,w),\rel}$, and defining a cardinality
constraint as $\pi(a)=|V_1|-k_1$, $\pi(b)=k_1$, $\pi(c)=k_2$,
$\pi(d)=|V_2|-k_2$. It is straightforward that for any solution $\vf$
of $\cP$ the set $S_\vf=\{v\in V\mid \vf(v)\in\{b,c\}\}$ is an
independent set, $S_\vf\cap V_1=\{v\mid \vf(v)=b\}$, $S_\vf\cap
V_2=\{v\mid \vf(v)=c\}$. Set $S_\vf$ satisfies the required conditions
if and only if $\vf$ satisfies $\pi$. Conversely, for any independent
set $S\sse V$ mapping $\vf$ given by
$$
\vf_S(v)=\left\{\begin{array}{ll}
a, & \text{if $v\in V_1\setminus S$},\\
b, & \text{if $v\in V_1\cap S$},\\
c, & \text{if $v\in V_2\cap S$},\\
d, & \text{if $v\in V_2\setminus S$},
\end{array}\right.
$$
is a solution of $\cP$ that satisfies $\pi$ if and only if $|S\cap V_1|=k_1$ and $|S\cap V_2|=k_2$.

\smallskip

{\sc Case 2.} $|\{a,b,c,d\}|=2$.

\smallskip

Then $\rel$ is a binary relation with 3 tuples in it over a 2-element set. By \cite{Creignou08:cardinality} $\CCSP(\rel)$ is NP-complete. 

\smallskip 

Thus in the remaining cases, we can assume that
$|\{a,b,c,d\}|=3$, and therefore $\{a,b\}\cap\{c,d\}\ne\eps$.  

\medskip

{\sc Claim 1.} One of the projections $\pr_1\rel$
or $\pr_2\rel$ contains only 2 elements. 

\smallskip

Let us assume the converse. Let $\pr_2\rel=\{c,d,x\}$,
$x\in\{a,b\}$ (as $\rel$ is over a 3-element set). We consider two
cases. Suppose  first $c\not\in\{a,b\}$ (implying $d\in\{a,b\}$).  If $(b,x)\not\in\rel$, then the
restriction of $\rel$ onto $\{a,b\}$ contains $(a,d),(b,d),(a,x)$, but
does not contain $(b,x)$. Thus $\rel_{|\{a,b\}}$ is not a thick mapping, a
contradiction with the minimality assumption. If $(b,x)\in\rel$ then the 
set $B=\{a,b\}=\{x\mid
(b,x)\in\rel\}$ is definable in $\rel$.  Observe that
$\rel'(x,y)=\rel(x,y)\wedge B(x)$ is not a thick mapping and definable
in $\rel$. A contradiction with the choice of $\rel$.

  Now suppose that $d\not\in\{a,b\}$ (implying $c\in\{a,b\}$). If $(b,x)\in\rel$, then the
  restriction $\rel_{|\{a,b\}}$ is not a thick mapping, as
  $(a,c),(a,x),(b,x)\in\rel$ and $(b,c)\not\in\rel$. Otherwise let
  $(b,x)\not\in\rel$. By the assumption made $|\pr_1\rel|=3$, that is,
  $d\in\pr_1\rel$. We consider 4 cases depending on whether $(d,c)$
  and $(d,x)$ are contained in $\rel$. If $(d,c),(d,x)\not\in\rel$, then, as $a\in\{x,c\}$,
  the relation $\rel_{|\{a,d\}}$ is not a thick mapping (recall that
  $(d,d)\in\rel$). If $(d,c),(d,x)\in\rel$, then we can restrict $\rel$
  on $\{d,b\}$ (note that $b\in \{c,x\}$). Finally, if $(d,c)\in\rel, (d,x)\not\in\rel$ [or
  $(d,x)\in\rel, (d,c)\not\in\rel$], then the relation $B=\{d,c\}$
  [respectively, $B=\{d,x\}$] is definable in $\rel$. It remains to
  observe that $\rel'(x,y)=\rel(x,y)\wedge B(x)$ is not a thick
  mapping. This concludes the proof of the claim.
  
  \smallskip

 Thus we can assume that one of the projections $\pr_1\rel$ or $\pr_2\rel$
  contains only 2 elements. Without loss of generality, let
  $\pr_1\rel=\{a,b\}$.  In the remaining cases, we assume
  $\pr_2\rel=\{c,d,x\}$, where $x\in \{a,b\}$ and $x$ may not be
  present.

\smallskip

{\sc Case 3.} 
Either
\begin{enumerate}[$\bullet$]
\item   $c\not\in\{a,b\}$ ({\sc Subcase 3a}), or
\item  $d\not\in\{a,b\}$ and $(b,x)\not\in\rel$ ({\sc Subcase 3b}). 
\end{enumerate}

\smallskip

In this case, given an instance $G=(V,E), V_1, V_2, k_1,k_2$ of BIS, we construct an instance $\cP=(V',\Cc)$ of $\CCSP(\rel)$ as follows.
\begin{enumerate}[$\bullet$]
\item
$V'=V_2\cup \displaystyle\bigcup_{w\in V_1} V^w$, where all the sets $V_2$ and $V^w$, $w\in V_1$, are disjoint, and $|V^w|=2|V|$.
\item
For any $(u,w)\in E$ the set $\Cc$ contains all constraints of the form $\ang{(v,w),\rel}$ where $v\in V^u$. 
\item
The cardinality constraint $\pi$ is given by the following rules:
\begin{enumerate}[$-$]
\item  Subcase 3a:\\
 $\pi(c)=k_2$, $\pi(a)=(|V_1|-k_1)\cdot 2|V|$,
  $\pi(b)=k_1\cdot 2|V|+(|V_2|-k_2)$ if $d=b$, and \\$\pi(c)=k_2$,
  $\pi(a)=(|V_1|-k_1)\cdot 2|V|+(|V_2|-k_2)$, $\pi(b)=k_1\cdot 2|V|$
  if $d=a$.
\item Subcase 3b:\\ $\pi(d)=|V_2|-k_2$, $\pi(a)=(|V_1|-k_1)\cdot 2|V|$,
  $\pi(b)=k_1\cdot 2|V|+k_2$ if $c=b$, and \\$\pi(d)=|V_2|-k_2$,
  $\pi(a)=(|V_1|-k_1)\cdot 2|V|+k_2$, $\pi(b)=k_1\cdot 2|V|$ if $c=a$.
\end{enumerate}
\end{enumerate}
If $G$ has a required independent set $S$, then consider a mapping
$\vf:V'\to D$ given by
$$
\vf(v)=\left\{\begin{array}{ll}
a, & \text{if $v\in V^w$  and $w\in V_1\setminus S$},\\
b, & \text{if $v\in V^w$ and $w\in V_1\cap S$},\\
c, & \text{if $v\in V_2\cap S$},\\
d, & \text{if $v\in V_2\setminus S$},
\end{array}\right.
$$
For any $\ang{(u,v),\rel}\in\Cc$, $u\in V^w$, either $w\not\in S$ or $v\not\in S$. In the first case $\vf(u)=a$ and so $(\vf(u),\vf(v))\in\rel$. In the second case $\vf(u)=b$ and $\vf(v)=d$. Again, $(\vf(u),\vf(v))\in\rel$. Finally it is straightforward that $\vf$ satisfies the cardinality constraint $\pi$.

Suppose that $\cP$ has a solution $\vf$ that satisfies $\pi$. Since
$\pr_1\rel=\{a,b\}$ and we can assume that $G$ has no isolated
vertices, for any $u\in V^w$, $w\in V_1$, we have $\vf(u)\in\{a,b\}$.
Also if for some $u\in V^w$ it holds that $\vf(u)=b$ and $\vf(v)=c$
for $v\in V_2$ then $(w,v)\not\in E$. We include into $S\sse V$ all
vertices $w\in V_1$ such that there is $u\in V^w$ with $\vf(u)=b$. By
the choice of the cardinality of $V^w$ and $\pi(b)$ there are at least
$k_1$ such vertices. In Subcase 3a, we include in $S$ all
vertices $v\in V_2$ with $\vf(v)=c$. There are exactly $k_2$ vertices
like this, and by the observation above $S$ is an independent set. In
Subcase 3b, we include in $S$ all vertices $v\in V_2$ with
$\vf(v)\in\{a,b\}$. By the choice of $\pi(d)$, there are at least $k_2$
such vertices. To verify that $S$ is an independent set it suffices to
recall that in this case $(b,x)\not\in\rel$, and so
$(b,a),(b,b)\not\in\rel$.

\smallskip

{\sc Case 4.} $d\not\in\{a,b\}$ and $(b,x)\in\rel$.

\smallskip

In this case $\{c,x\}=\{a,b\}$ and $(a,c),(a,x),(b,x)\in\rel$ while $(b,c)\not\in\rel$. 
Therefore $\rel_{|\{a,b\}}$ is not a thick mapping. A contradiction with the choice of $\rel$.
\end{proof}

Next we show hardness in the case when  there is a  crossing pair of equivalence
relations.
With a simple observation, we can obtain
a binary relation that is not a thick mapping and then apply
Lemma \ref{lem:non-thick-mapping}.

\begin{lemma}\label{lem:non-non-crossing}
Let $\al,\beta$ be a  crossing pair of equivalence relations. Then
$\CCSP(\{\al,\beta\})$ is NP-complete.   
\end{lemma}

\begin{proof}
Let $\al,\beta$ be equivalence relations on the same domain $D$.
This means that there are
$a,b,c\in D$ such that $\ang{a,c}\in\al\setminus\beta$ and
$\ang{c,b}\in\beta\setminus\al$. Let $\al'=\al_{|\{a,b,c\}}$ and
$\beta'=\beta_{|\{a,b,c\}}$. Clearly, 
\begin{eqnarray*}
&& \al'=\{(a,a),(b,b),(c,c),(a,c),(c,a)\},\\ 
&& \beta'=\{(a,a),(b,b),(c,c),(b,c),(c,b)\}.
\end{eqnarray*}
By Lemma~\ref{lem:restriction}, $\CCSP(\{\al',\beta'\})$ is
polynomial-time reducible to $\CCSP(\{\al,\beta\})$. Consider the
relation $\rel=\al'\circ\beta'$, that is, $\rel(x,y)=\exists z
(\al'(x,z)\wedge\beta'(z,y))$. We have that $\CCSP(\rel)$ is
reducible to $\CCSP(\{\al',\beta'\})$ and
$$
\rel=\{(a,a),(b,b),(c,c),(a,c),(c,a),(b,c),(c,b),(a,b)\}.
$$
Observe that $\rel$ is not a thick mapping and by
Lemma~\ref{lem:non-thick-mapping}, $\CCSP(\rel)$ is NP-complete. 
\end{proof} 

Finally, we prove hardness in the case when there is a relation that
is not 2-decomposable. An example of such a relation is a ternary
Boolean affine relation $x+y+z=c$ for $c=0$ or $c=1$.
The CSP with global cardinality
constraints for this relation is NP-complete by
\cite{Creignou08:cardinality}. Our strategy is to obtain such a relation
from a relation that is not 2-decomposable. However, as in
Lemma~\ref{lem:non-thick-mapping}, we have to consider several cases.

We start with the following simple lemma:
\begin{lemma}\label{lem:classes}
  Let $\al$ be an equivalence relation on a set $D$ and $a\in D$.
 Then $a^\al\in\dang{\al,C_a}'$.
\end{lemma}

\begin{proof}
The unary relation $R(x)=\exists y((\ang{x,y}\in \alpha) \wedge C_a(y))$ is equal
to $a^\al$ and is clearly in $\dang{\al,C_a}'$.
\end{proof}

\begin{lemma}\label{lem:non-2-decomposable}
Let $\rel$ be a relation that is not 2-decomposable. Then
$\CCSP(\{\rel\})$ is NP-complete. 
\end{lemma}

\begin{proof}
  We can assume that every binary relation in
  $\dang{\{\rel\}\cup\{C_a\mid a\in D\}}'$ is a thick mapping,
  and no pair of equivalence relations from this set cross,
  otherwise the problem is NP-complete by
  Theorems~\ref{the:decision-pp-defin}, ~\ref{the:adding-constants},
  and Lemmas~\ref{lem:non-thick-mapping},~\ref{lem:non-non-crossing}.
  Furthermore, we can choose
  $\rel$ to be the `smallest' non-2-decomposable relation in the sense
  that every relation obtained from $\rel$ by restricting on a proper
  subset of $D$ is 2-decomposable, and every relation
  $\rel'\in\dang{\{\rel\}\cup\{C_a\mid a\in D\}}'$ that either have
  smaller arity, or $\rel'\subset\rel$, is also 2-decomposable.

We claim that relation $\rel$ is ternary. Indeed, it cannot be binary by assumptions
made about it. Suppose that $\ba\not\in\rel$  
is a tuple such that $\pr_{ij}\ba\in\pr_{ij}\rel$ for any $i,j$. Let 
\begin{eqnarray*}
\rel'(x,y,z) &=& \exists x_4\zd x_n (\rel(x,y,z,x_4\zd x_n)\wedge\\
& &\quad C_{\ba[4]}(x_4)\wedge\ldots \wedge C_{\ba[n]}(x_n)).
\end{eqnarray*}
It is straightforward that
$(\ba[1],\ba[2],\ba[3])\not\in\rel'$, while, since any proper
projection of $\rel$ is 2-decomposable, $\pr_{\{2\zd
  n\}}\ba\in\pr_{\{2\zd n\}}\rel$, $\pr_{\{1,3\zd
  n\}}\ba\in\pr_{\{1,3\zd n\}}\rel$, $\pr_{\{1,2,4\zd
  n\}}\ba\in\pr_{\{1,2,4\zd n\}}\rel$, implying 
$(\ba[2],\ba[3])\in\pr_{23}\rel'$, $(\ba[1],\ba[3])\in\pr_{13}\rel'$,
$(\ba[1],\ba[2])\in\pr_{12}\rel'$, respectively. Thus $\rel'$ is not 2-decomposable,
a contradiction with assumptions made.

Let $(a,b,c)\not\in\rel$ and $(a,b,d),(a,e,c),(f,b,c)\in\rel$, and let $B=\{a,b,c,d,e,f\}$.  
As $\rel_{|B}$ is not 2-decomposable, we should have $\rel=\rel_{|B}$.

If $\rel_{12}=\pr_{12}\rel$ is a thick mapping with respect to some 
equivalence relations $\eta_{12},\eta_{21}$ (see p.~\ref{page:with-respect}), 
$\rel_{13}=\pr_{13}\rel$ is a thick mapping
with respect to $\eta_{13},\eta_{31}$, and $\rel_{23}=\pr_{23}\rel$ is
a thick mapping with respect to $\eta_{23},\eta_{32}$, then
$\ang{a,f}\in\eta_{12}\cap\eta_{13}$,
$\ang{b,e}\in\eta_{21}\cap\eta_{23}$, and
$\ang{c,d}\in\eta_{31}\cap\eta_{32}$. Let the corresponding classes of
$\eta_{12}\cap\eta_{13}$, $\eta_{21}\cap\eta_{23}$, and
$\eta_{31}\cap\eta_{32}$ be $B_1,B_2$, and $B_3$, respectively. Then
$B_1=\pr_1\rel$, $B_2=\pr_2\rel$, $B_3=\pr_3\rel$. Indeed, if one of
these equalities is not true, since by Lemma~\ref{lem:classes}
the relations $B_1,B_2,B_3$ are pp-definable in $\rel$ without equalities, the
relation $\rel'(x,y,z)=\rel(x,y,z)\wedge B_1(x)\wedge B_2(y)\wedge
B_3(z)$ is pp-definable in $\rel$ and the constant relations, is
smaller than $\rel$, and is not 2-decomposable. 

Next we show that $(a,g)\in\pr_{12}\rel$ for all $g\in\pr_2\rel$. If
there is $g$ with $(a,g)\not\in\pr_{12}\rel$ then setting
$C(y)=\exists z (\pr_{12}\rel(z,y)\wedge C_a(z))$ we have $b,e\in C$
and $C\ne\pr_2\rel$. Thus $\rel'(x,y,z)=\rel(x,y,z)\wedge C(y)$ is
smaller than $\rel$ and is not 2-decomposable. The same is true for
$b$ and $\pr_2\rel$, and for $c$ and $\pr_3\rel$. Since every binary
projection of $\rel$ is a thick mapping this implies that
$\pr_{12}\rel=\pr_1\rel\tm\pr_2\rel$,
$\pr_{23}\rel=\pr_2\rel\tm\pr_3\rel$, and
$\pr_{13}\rel=\pr_1\rel\tm\pr_3\rel$. 

For each $i\in\{1,2,3\}$ and every $p\in\pr_i\rel$, the relation
$\rel^p_i(x_j,x_k)=\exists x_i(\rel(x_1,x_2,x_3)\wedge C_p(x_i))$,
where $\{j,k\}=\{1,2,3\}\setminus\{i\}$, is definable in $\rel$ and
therefore is a thick mapping with respect to, say, $\eta^p_{ij},\eta^p_{ik}$.
Our next step is to show that $\rel$ can be chosen
such that $\eta^p_{ij}$ does not depend on the choice of $p$ and $i$.

If one of these relations, say, $\rel^p_1$, equals
$\pr_2\rel\tm\pr_3\rel$, while another one, say $\rel^q_1$ does not,
then consider $\rel^c_3$. We have $\{p\}\tm\pr_2\rel\sse\rel^c_3$.
Moreover, since by the choice of $\rel$ relation $\rel^q_1$ is a
non-trivial thick mapping there is $r\in\pr_2\rel$ such that
$(r,c)\not\in\rel^q_1$, hence $(q,r)\not\in\rel^c_3$. Therefore
$\rel^c_3$ is not a thick mapping, a contradiction. 
Since $\rel^a_1$ does not equal $\pr_2\rel\tm\pr_3\rel$ 
(and $\rel^b_2\ne\pr_1\rel\tm\pr_3\rel$, $\rel^c_3\ne\pr_1\rel\tm\pr_2\rel$), 
every $\eta^p_{ij}$ is non-trivial. Observe that due to the equalities 
$\pr_{ij}\rel=\pr_i\rel\tm\pr_j\rel$, we also have that the unary 
projections of $\rel^p_i$ are equal to those of $\rel$ for any $p$; 
and therefore all the equivalence relations $\eta^p_{ji}$, for a 
fixed $j$, are on the same domain, $\pr_j\rel$.
Let
$$
\eta_i=\bigvee_{\substack{j\in\{1,2,3\}\setminus\{i\}\\
    p\in\pr_j\rel}} \eta^p_{ji}.
$$
Since for any non-crossing $\al,\beta$ we have $\al\join\beta=\al\circ\beta$, the relation $\eta_i$ is
pp-definable in $\rel$ and constant relations without
equalities. Since all the $\eta^p_{ji}$ are non-trivial, $\eta_i$ is
also non-trivial. We set  
\begin{eqnarray*}
\rel'(x,y,z) &=& \exists x',y',z'(\rel(x',y',z')\wedge\eta_1(x,x')\wedge\\
& &\quad \eta_2(y,y')\wedge\eta_3(z,z')).
\end{eqnarray*}
Let $\relo^p_i$ be defined for $\rel'$ in the same way as $\rel^p_i$
for $\rel$. Observe that since $(p,q,r)\in\rel'$ if and only if there
is $(p',q',r')\in\rel$ such that $\ang{p,p'}\in\eta_1$,
$\ang{q,q'}\in\eta_2$, $\ang{r,r'}\in\eta_3$, the relations
$\relo^p_1$, $\relo^q_2$, $\relo^r_3$ for $p\in\pr_1\rel',
q\in\pr_2\rel', r\in\pr_3\rel'$ are thick mappings with respect to the
equivalence relations $\eta_2,\eta_3$, relations $\eta_1,\eta_3$, and
relations $\eta_1,\eta_2$, respectively. All the binary
projections of $\rel'$ equal to the direct product of the
corresponding unary projections, while $\eta_1,\eta_2,\eta_3$ are
non-trivial, which means $\rel'$ is not the direct product of its unary 
projections, and therefore it is not 2-decomposable. We then can 
replace $\rel$ with $\rel'$. Thus we have achieved that  $\eta^p_{ij}$ 
does not depend on the choice of $p$ and $i$. 

Next we show that $\rel$ can be chosen such that
$\pr_1\rel=\pr_2\rel=\pr_3\rel$, $\eta_1=\eta_2=\eta_3$, and for each
$i\in\{1,2,3\}$ there is $r\in\pr_i\rel$ such that $\rel^r_i$ is a
reflexive relation. If, say, $\pr_1\rel\ne\pr_2\rel$, or
$\eta_1\ne\eta_2$, or $\rel^r_3$ is not reflexive for any
$r\in\pr_3\rel$, consider the following relation 
$$
\rel'(x,y,z)=\exists y',z' (\rel(x,y',z)\wedge \rel(y,y',z')\wedge C_d(z')).
$$
First, observe that $\pr_{ij}\rel'=\pr_i\rel'\tm\pr_j\rel'$ for any
$i,j\in\{1,2,3\}$. Then, for any fixed $r\in\pr_3\rel'=\pr_3\rel$ the
relation $\relo^r_3=\{(p,q)\mid (p,q,r)\in\rel'\}$ is the product
$\rel^r_3\circ(\rel^d_3)^{-1}$, that is, a non-trivial thick mapping.
Thus $\rel'$ is not 2-decomposable. Furthermore,
$\pr_1\rel'=\pr_2\rel'=\pr_1\rel$, for any $r\in\pr_3\rel'$ the
relation $\relo^r_3$ is a thick mapping with respect to
$\eta_1,\eta_1$, and $\relo^d_3$ is reflexive. Repeating this
procedure for the two remaining pairs of coordinate positions, we obtain a
relation $\rel''$ with the required properties. Observe that repeating 
the procedure does not destroy the desired properties where they 
are already achieved.
Replacing $\rel$ with $\rel''$ if necessary, we may assume that $\rel$
also has all these properties. 

Set $B=\pr_1\rel=\pr_2\rel=\pr_3\rel$ and
$\eta=\eta_1=\eta_2=\eta_3$. Let $p\in B$ be such that $\rel^p_1$ is
reflexive. Let also $q\in B$ be such that $\ang{p,q}\not\in\eta$. Then
$(p,p,p),(p,q,q)\in\rel$ while $(p,p,q)\not\in\rel$. Choose $r$ such
that $(r,p,q)\in\rel$. Then the restriction of $\rel$ onto 3-element
set $\{p,q,r\}$ is not 2-decomposable. Thus $\rel$ can be assumed to
be a relation on a 3-element set.  

If $\eta$ is not the equality relation, say, $\ang{p,r}\in\eta$, then
as the restriction of $\rel$ onto $\{p,q\}$ is still a not
2-decomposable relation, $\rel$ itself is a relation on the set
$\{p,q\}$. Identifying $p$ and $q$ with 0 and 1 it is not hard to 
see that it is the affine relation $x+y+z=0$ on
$\{p,q\}$. The CSP with global cardinality constraints for this
relation is NP-complete by \cite{Creignou08:cardinality}. 

Suppose that $\eta$ is the equality relation. Since each of
$\rel^p_1,\rel^q_1,\rel^r_1$ is a mapping and
$\rel^p_1\cup\rel^q_1\cup\rel^r_1=B^2$, it follows that the three
relations are disjoint. As $\rel^p_1$ is the identity mapping,
$\rel^q_1$ and $\rel^r_1$ are two cyclic permutations of (the
3-element set) $B$. Hence either $(p,q)$ or $(q,p)$ belongs to
$\rel^q_1$. Let it be $(p,q)$.  Restricting $\rel$ onto $\{p,q\}$ we
obtain a relation $\rel'$ whose projection $\pr_{23}\rel'$ equals
$\{(p,p),(q,q),(p,q)\}$, which is not a thick mapping. A contradiction
with the choice of $\rel$.
\end{proof}

%%%%%%%%%%%%%%%%%%%%%%%%%%%%%%%%%%%%%%%%%%%%%%%%%%%%
%%%%%%%%%%%%%%%%%%%%%%%%%%%%%%%%%%%%%%%%%%%%%%%%%%%%
\bibliographystyle{plain}  
%% \bibliography{global}

\vspace{-20 pt}

\end{document}